\documentclass[journal,onecolumn,11pt]{IEEEtran}
\pdfoutput=1
\usepackage[draft, blank]{ieeefig}
\usepackage{cite}
\usepackage{amsmath}
\usepackage{amsthm}
\usepackage{amsfonts}
\usepackage{amssymb}
\usepackage{bm}
\usepackage{bbm}
\usepackage{dsfont}
\usepackage{graphicx}
\usepackage{subcaption}
\usepackage{color}

\newtheorem{theorem}{Theorem}
\newtheorem{lemma}{Lemma}
\newtheorem{corollary}{Corollary}

\newtheorem{assumption}{Assumption}

\theoremstyle{remark}

\newtheorem*{remark*}{Remark}

\newcommand{\abs}[1]{\left|#1\right|}
\newcommand{\mop}[1]{\mathcal{#1}}
\newcommand{\opn}[1]{\operatorname{#1}}
\newcommand{\vect}[1]{\mathbf{#1}}
\newcommand{\norm}[1]{\left\|#1\right\|}

\newcommand{\epf}{\bar{\varepsilon}}
\newcommand{\epo}{\bar{\varepsilon}_{sc}}
\newcommand{\epot}{\tilde{\bar{\varepsilon}}_{sc}}
\newcommand{\epn}{\varepsilon_0}

\newcommand{\ip}[2]{\langle #1, #2 \rangle}

\newcommand*{\from}{\colon}

\newcommand{\mrm}[1]{\mathrm{#1}}

\DeclareMathAlphabet{\bi}{OML}{cmm}{b}{it}
\DeclareMathAlphabet{\bcal}{OMS}{cmsy}{b}{n}
\DeclareMathAlphabet{\brmn}{OT1}{cmr}{bx}{n}

\DeclareMathSymbol{\R}{\mathalpha}{AMSb}{"52}

\newcommand{\n}{\nabla}

\newcommand{\A}{\alpha}

\usepackage{amsthm}
\usepackage{amsfonts}

\newcommand{\D}{d}

\def \x{\mathbf{x}}
\def \n{\mathbf{n}}
\def \y{\mathbf{y}}
\def \u{\mathbf{u}}

\def \r{\mathbf{r}}
\def \z{\mathbf{z}}
\def \s{\mathbf{s}}
\def \v{\mathbf{v}}

\def \b{\mathbf{b}}
\def \w{\mathbf{w}}

\def \a{\mathbf{a}}
\def \e{\mathbf{e}}
\def \c{\mathbf{c}}
\def \d{\mathbf{d}}

\def \R{\mathbf{R}}
\def \D{\mathbf{D}}

\def \K{\mathbf{K}}

\def \A{\mathbf{A}}

\def \0{\mathbf{0}}

\begin{document}
\title{Passive Polarimetric Multistatic Radar Detection of Moving
  Targets}
\author{\IEEEauthorblockN{Il-Young Son and Birsen
    Yaz{\i}c{\i}$^*$,~\IEEEmembership{Senior Member, IEEE}}%
  \thanks{I. Son and B. Yaz{\i}c{\i} are with the Department of
    Electrical, Computer, and Systems Engineering, Renssealer
    Polytechnic Institute, Troy, NY 12180, USA}
\thanks{This material is based upon work supported by the Air Force Office of Scientific Research
  (AFOSR) under award number FA9550-16-1-0234,
  and by the National Science Foundation (NSF) under Grant No. CCF-1421496.}
\thanks{* Corresponding author.}
}

\IEEEtitleabstractindextext{
  \begin{abstract}
    We study the exploitation of polarimetric diversity in passive multistatic radar for detecting
    moving targets.  We first derive a data model that takes into account
    polarization and anisotropy of targets inherent in multistatic configurations. Unlike
    conventional isotropic models in which targets are modeled as a collection of uniform spheres,
    we model targets as a collection of dipole antennas with unknown directions. We consider a
    multistatic configuration in which each receiver is equipped with a pair of orthogonally
    polarized antennas, one directed to a scene of interest collecting target-path signal and
    another one having a direct line-of-sight to a transmitter-of-opportunity collecting
    direct-path signal.
    We formulate the detection of moving target problem in a generalized likelihood ratio test
    framework under the assumption that direct-path signal is available. We show that the result can be
    reduced to the case in which the direct-path signal is absent. We present a method for estimating
    the dipole moments of targets. Extensive numerical simulations show the performance of both the
    detection and the dipole estimation tasks with and without polarimetric diversity.
\end{abstract}

}

\maketitle
\IEEEdisplaynontitleabstractindextext

\section{Introduction}
\subsection{Overview}
Passive radar senses the environment using transmitters of opportunity such as digital television
stations, wideband cell-phone towers, radio transmissions, satellites, etc.  The past decade has
seen a sharp rise in the availability of such transmitters and with it a growing research interest
in passive radar
systems~\cite{wang09,wang10,wang12,LWang12,Wang11,Baker05_1,baker05_2,Dawidowicz12,Kulpa12,Wacks14,Yarman10,mason15}.



The polarimetric diversity in the context of passive radar, refers to each receiver being equipped
with a pair of linear and orthogonally polarized antennas.
Polarimetric diversity in passive radar is important for the following reasons:
\begin{itemize}
\item In conventional passive radars,
  sources,
  scatterers and receivers are typically assumed to be isotropic and polarimetric diversity is not
  considered~\cite{son2007radar,wang10,LWang12,Wang11,hack12,hack14,hack2014centralized,bialkowski2011generalized,palmer2013dvb,colone2011direction,cui14}.
  Under the isotropy assumption, the vector wave equation simplifies to a scalar one.
  However, for the multistatic
  configuration
  the directionality/polarization of electomagnetic (EM) waves from the scatterers
  and antennas become increasingly important and the isotropy assumption
  may no longer hold.
\item
  Polarization diversity provides an additional antenna at each receiver
  that is polarized in an orthogonal direction. This means that even if the signal received is weak
  at one antenna, a relatively strong signal is received at the other ensuring the effectiveness of
  spatial diversity offered by the multistatic configuration.
\item Polarimetric radar can provide additional information that is not available in conventional
  non-polarimetric radars, including polarization characteristics of targets that can aid in
  classiﬁcation and recognition tasks~\cite{webster12,lee09, cloude96,
    boerner07,moreira13, jackson09}.
\end{itemize}
These reasons motivate the exploitation of polarimetric diversity in passive radar to produce
better detection, imaging and classification performance than that of conventional non-polarimetric
radar.


In this paper, we consider a multistatic configuration in which each receiver is equipped with a
pair of polarimetrically diverse antennas, one of which has a direct line-of-sight to a transmitter
of opportunity collecting direct-path signal and another one is directed to a scene of interest
collecting target-path signal. We derive a target model and a received signal model for the
scattered field from moving
targets in a
polarimetrically diverse configuration starting from first principle.  Our derivation is based on
the vector wave equation, using the dyadic Green's function and dyadic reflectivity taking
into account the linear motion of a moving target.  We reduce the dyadic reflectivity to a spatially
distributed dipole target model.  This model is derived from the eigendecomposition of the target
dyad. This decomposition induces a three colocated dipoles for the target at each spatial location with
each eigenvector interpreted as a dipole moment.  To reduce the size of the problem, we consider
only the dominant eigenvalue/eigenvector - effectively modeling a scatterer as a single dipole with
unknown dipole direction and reflectivity.
We then model the receive and transmit antennas as dipole antennas to arrive at the final form of
the
received signal model.

We assume that the transmitted waveform and transmit antenna polarization state is unknown and formulate the moving target detection problem as a generalized likelihood ratio test. 
We derive the test statistic for both
the case in which the direct-path signal is available and one in which it is not. We present a
method of estimating target dipole direction and hence its polarization state. In
  \cite{son17analysis}, we analyze the performance gains in moving target detection due to polarimetric
  diversity in passive multistatic radar.
Extensive numerical simulations show that polarimetric diversity provides superior detection
performance than that of non-diverse case. Additionally, it provides information on anistoropic
characteristic of scatterers that can be used for target recognition tasks.
\subsection{Related Work}
The exploitation of polarimetry in radar application has its origins as far back as the mid 20th
century~\cite{sinclair48,sinclair50,kennaugh52,deschamps51,graves56}. However traditionally, radar polarimetry has been applied to active monostatic
systems, as in polarimetric Synthetic Aperture Radar (SAR) and polarimetric SAR
interferometry~\cite{moreira13,boerner07,voccola13,voccola11}.
In~\cite{webster2014b} multistatic polarimetric active radar has been studied for imaging of moving
targets.
In
this work, the multistatic system exhibits polarimetric diversity in both transmitter and receiver
which are both modeled as dipole antennas. The imaging scheme, however, does not extract
polarization information about the target scene and consists of weighted sum of correlations
between the received signal and scaled and delayed versions of the transmitted waveforms. In
contrast, we study multistatic passive radar systems in which a priori knowledge on the transmitted
waveforms or the polarimetric state of the transmitters of opportunity is not
available. Furthermore, our objective is not only to estimate the reflectivity but also the
polarization state of targets.

The dipole target model was first studied
in~\cite{gustafsson04}.  In~\cite{voccola13},
this model is utilized and a filtered-backprojection type imaging scheme for monostatic active
polarimetric SAR was developed.~\cite{voccola13} differs from our work in that it considers imaging of static
targets using an active, monostatic conﬁguration. In addition, it constrains target dipole moments
to lie on a flat ground plane. We do not assume such constraints on the dipole moment direction.

Multistatic passive radar detection problems have been studied
in~\cite{wang10,wang12,LWang12,hack14,hack2014centralized,bialkowski2011generalized}.  These works do
not consider polarimetric diversity nor the estimation of target's polarization state.
In~\cite{wang10}, spatially resolved detection of stationary
scatterers is studied for
multiple scattering environment.  This is extended to moving targets
in~\cite{wang12,LWang12}.  In these papers, the authors make explicit isotropic assumptions for the target,
transmitters and receivers.
Our GLRT-based detection scheme is
similar to the one presented in~\cite{hack14,hack2014centralized}.  However, in addition to our target and data models
being different, we derive our test-statistic by considering the full continuous data and noise
processes.  We also relax the assumption that
all noise processes have a common variance.

Passive polarimetry has been
explored in remote sensing applications in the
optical and infrared spectrum~\cite{tyo06}. In acoustic imaging, ambient polarized acoustic noise
field is
used to recover the scattering matrix~\cite{wapenaar13}.
In~\cite{colone2014,colone2015}, polarimetric diversity for passive radar detection (PCL) is
studied.
In this work, they experimentally assess the effectiveness of polarimetric diversity in mitigating
interference on the direct-path signal.  Not only is their target, data model and test-statistic different
from present work, but they do not consider estimation of the target's polarization state.  In
addition, we consider the case in which direct-path signal is unavailable.
In~\cite{son15}, we studied multistatic polarimetric passive radar for a single or widely spatially
separated stationary targets without direct-path signals. In the present work, we
consider ground moving targets and derive spatially resolved target detection scheme along with
estimation of its polarization state.
In addition, we consider two cases: when the direct-path
signal is available at each receiver and when no such signal is available.
\subsection{Organization of the Paper}
The paper is organized as follows: in Section~\ref{sec:forward} we derive the model of scattered
vector field from a moving target.  In Section~\ref{sec:dipoletarg}, we
derive the dipole target model from the eigendecomposition of dyadic permittivity.  The
polarimetric data model is derived in Section~\ref{sec:data_model}.
In Section~\ref{sec:glrt}, we use the polarimetric data model to derive the target detection and target dipole moment estimation problems.
In Section~\ref{sec:simulations}, we present numerical simulations to demonstrate the performance
of the detection and estimation methods.  
 Section~\ref{sec:conclusions} concludes the paper.
\section{Incident and Scattered Field Models}\label{sec:forward}
\subsection{Geometric Configuration of Multistatic Passive Radar}
Fig.~\ref{fig:configuration} depicts a typical configuration of a distributed passive radar that is of interest in this paper.
We assume that there are $M$ stationary receivers distributed about a scene of interest and a
single transmitter of opportunity.  The location of the $k$-th receiver is denoted as
$\a^r_k\in\mathbb{R}^3$ and the location of the transmitter of opportunity is denoted as $\a^t\in\mathbb{R}^3$. We assume that
each receiver exhibits polarimetric diversity.
In other words, each receiver consists of a  pair of orthogonal dipole antennas
receiving scattered signals from the scene of interest.  Furthermore, we consider a case in which at each
receiver, the \emph{direct-path}\footnote{\emph{Direct-path} signal is the signal received
directly from the transmitter without any background scattering.} signal and
\emph{target-path}\footnote{\emph{Target-path} travels from the transmitter to the target and then
  scatters from the target and travels to the receiver.} signal is separated.
This
separation can be achieved, for instance, by employing a beamforming technique on an array of receive
antennas at each receiver location as in~\cite{hack14}.
 The direct-path signal is also referred to as the \emph{reference} channel and the
 target-path signal is referred to as the \emph{surveillance} channel in the literature~\cite{hack14,cui14}.
We assume that the transmitter consists of a single dipole antenna and that the dipole direction of
the transmitter is not known.
\begin{figure}[!ht]
  \centering
  \includegraphics[width=3in]{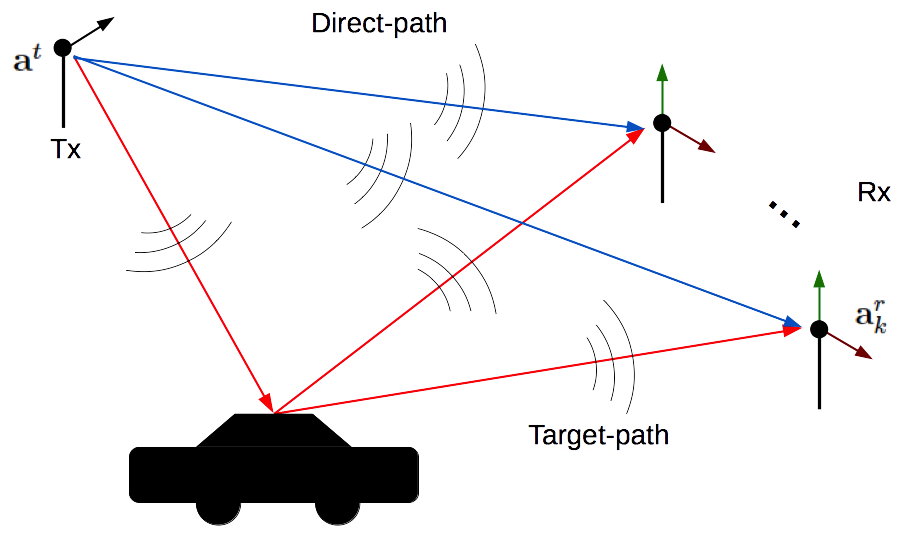}
  \caption{Distributed Polarimetric Passive Radar Scenario.  $\a^r_k$ is the $k$-th receiver
    location and $\a^t$ is the transmitter location.}\label{fig:configuration}
\end{figure}

\subsection{Scattered Field Model with Dyadic Reflectivity}
Since we are interested in modeling the directionality of EM waves scattered from the target of
interest, we begin with vector wave equation.  We assume that the medium is reciprocal and lossless
and is characterized by constant real-valued permeability, $\mu_0$ and spatially varying
real-valued dyadic permittivity, $\epf\from\mathbb{R}^3\to\mop{S}^3$ where $\mop{S}^3$ is the
space of $3\times
3$ real symmetric matrices~\cite{chew99}.  The vector wave equation is then~\cite{chew99}
\begin{equation}
  \nabla\times\nabla\times\mop{E} =
  \omega^2\mu_0\epf\mop{E} + i\omega\mu_0\mop{J}\label{eq:weq1}
\end{equation}
where $\mop{E}$ is the electric field, $\mop{J}$ is the source term and $\omega$ is the temporal
frequency.

\begin{assumption}\label{as:isotropic}
  We assume that the background medium is homogeneous and
  isotropic with
  scalar permittivity $\epn\in\mathbb{R}$.
\end{assumption}

Under Assumption~\ref{as:isotropic}, we model $\epf$ as a perturbation of $\epn$, i.e.,
\begin{equation}
  \epf(\x) = \epn I_3 + \epo(\x)
\end{equation}
where $I_3$ is a $3\times 3$ identity matrix.  We assume that $\epo$ is compactly supported.


We can then rewrite~\eqref{eq:weq1} as
\begin{equation}
  \nabla\times\nabla\times\mop{E} - k^2\mop{E} =
  \omega^2\mu_0\epo\mop{E} + i\omega\mu_0\mop{J}\label{eq:weq2}
\end{equation}
where $k = \omega\sqrt{\mu_0\epn}$ is the wavenumber.
Here, we note that $\mop{E} = \mop{E}^{in} + \mop{E}^{sc}$ where $\mop{E}^{in}$ is the incident
field originating from the source $\mop{J}$ and $\mop{E}^{sc}$ is the scattered field.  Using this
fact, we can split~\eqref{eq:weq2} into two vector wave equations as
\begin{equation}
  \nabla\times\nabla\times\mop{E}^i - k^2\mop{E}^i = i\omega\mu_0\mop{J}\label{eq:weqi}
\end{equation}
and
\begin{equation}
  \nabla\times\nabla\times\mop{E}^{sc} - k^2\mop{E}^{sc} = \omega^2\mu_0\epo\mop{E}.\label{eq:weqsc1}
\end{equation}
The solutions to the  vector wave equations in~\eqref{eq:weqi} and~\eqref{eq:weqsc1} are given by~\cite{chew99}
\begin{equation}
  \mop{E}^{in}(\mathbf{x},\omega) = i\omega\mu_0\int
  \hat{G}(\mathbf{x},\mathbf{x}',\omega)\mop{J}(\mathbf{x}',\omega)d\mathbf{x}'\label{eq:Ei}
\end{equation}
and
\begin{equation}
  \mop{E}^{sc}(\mathbf{y},\omega) = \omega^2\mu_0\int \hat{G}(\mathbf{y},\mathbf{x},\omega)\epo(\mathbf{x})\mop{E}(\mathbf{x},\omega)d\mathbf{x}\label{eq:Esc1}
\end{equation}
where  $\hat{G}(\y,\x,\omega)$ is the dyadic Green's function\footnote{In~\eqref{eq:dyadicG}, $c_0$
  is the wave speed in free-space and $\nabla\nabla$ is the Hessian operator.} given
by~\cite{chew99,voccola13}
\begin{equation}
  \hat{G}(\mathbf{y},\mathbf{x},\omega) = \left(I + \frac{\nabla\nabla}{k^2}\right)\frac{\mrm{e}^{\mrm{i}\omega\abs{\mathbf{y}-\mathbf{x}}/c_0}}{4\pi\abs{\mathbf{y}-\mathbf{x}}}.\label{eq:dyadicG}
\end{equation}
\eqref{eq:Esc1} is the vector version of the well-known Lippmann-Schwinger
equation~\cite{colton98}.  We assume a weak surface scattering and make the Born approximation
to linearize~\eqref{eq:Esc1}, arriving at
\begin{equation}
  \mop{E}^{sc}(\mathbf{y},\omega) \approx \omega^2\mu_0\int \hat{G}(\mathbf{y},\mathbf{x},\omega)\epo(\mathbf{x})\mop{E}^{in}(\mathbf{x},\omega)d\mathbf{x}.\label{eq:Esc2}
\end{equation}

\subsection{Scattered and Incident Field From a Moving Target}

In the previous section, we derived the general form of the scattered EM vector field.   In
this
section, we derive the model for the scattered field from a moving scatterer based
on~\eqref{eq:Esc1}.  We begin with the following assumption.

\begin{assumption}
  We assume that the target at $\x\in\mathbb{R}^3$ at time $t=0$ is moving with constant velocity $\v_{\x}\in\mathbb{R}^3$.  Let $\z = \x + \v_{\x}t$ be
  location of the target at time $t$.
\end{assumption}

From~\eqref{eq:Esc1}, we express the scattered field from the moving target in time domain as
\begin{equation}
  \begin{aligned}
  \mop{E}^{sc}(\mathbf{y},t) \approx \int
  &G(\mathbf{y}, \mathbf{x}+\v\tau,t-\tau)\mop{Q}(\mathbf{x},\v)\\
  &\partial^2_{\tau}\mop{E}^{in}(\mathbf{x}+\v\tau,\tau)d\tau d\mathbf{x}d\v
  \end{aligned}\label{eq:Esc3d}
\end{equation}
where $G(\y,\x,t)$ is the inverse Fourier transform of~\eqref{eq:dyadicG} and
\begin{equation}
  \mop{Q}(\x,\v) = \mu_0\epo(\x)\delta(\v-\v_{\x}).
\end{equation}
The model in~\eqref{eq:Esc3d} is similar to the one derived
in~\cite{wang12} for the scalar scattered field from a moving target.  We have extended this to the
vector field case using the solution to the vector wave equation.
This  model captures anisotropic scattering from targets in multistatic configurations.  In addition, consideration of the vector scattered field allows us to
model the polarimetry of the scattered field.

We assume that the ground topography $\psi\from\mathbb{R}^2\to\mathbb{R}$ is known.  Let a target
be initially located at $\x = [\bi{x}, \psi(\bi{x})]^T$ and moving with velocity $\v_{\x} = [\bi
v_{\x},\nabla_{\bi x}\psi(\bi x)\cdot\bi v_{\x}]^T$ where
$\bi{x},\bi{v}_{\x}\in\mathbb{R}^2$. Thus, we can write the reflectivity function, $\mop{Q}$, that depends on $\x, \v\in\mathbb{R}^3$ in
terms of $\tilde{\mop{Q}}$ that depends on $\bm{x},\bm{v}\in\mathbb{R}^2$, as follows:
\begin{equation}
  \mop{Q}(\x,\v) = \tilde{\mop{Q}}(\bi{x},\bi{v})\delta(x_3-\psi(\bi{x}))\delta(v_3-\nabla\psi(\bi x)\cdot\bi v_{\x})\label{eq:Q}
\end{equation}
where
\begin{equation}
  \tilde{\mop{Q}}(\bm{x},\bm{v}) =\mu_0\epot(\bm{x})\delta(\bm{v} -
  \bm{v}_{\x}).\label{eq:Qtilde}
\end{equation}
We refer to $\tilde{\mop{Q}}$ as the \emph{dyadic phase space reflectivity function}.  Note that $\epot$ is equal to the dyadic permittivity, $\epo$, with
the third spatial variable $x_3$ constrained to be $x_3 = \psi(\bm{x})$.  We
approximate~\eqref{eq:Qtilde} as
\begin{equation}
  \tilde{\mop{Q}}(\bm{x},\bm{v}) \approx \mu_0\epot(\bm{x})\varphi(\bm{v} - \bm{v}_{\x})
\end{equation}
where $\varphi$ is a smooth function approximating $\delta$, such as a Gaussian with a small
variance.
Inserting~\eqref{eq:Q} into~\eqref{eq:Esc3d} and integrating out the third component of $\x$ and
$\v$, we get
\begin{equation}
  \begin{aligned}
  \mop{E}^{sc}(\mathbf{y},t) \approx \int
  &G(\mathbf{y},\x+\v\tau,t-\tau) \tilde{\mop{Q}}(\bi{x},\bi{v})\\
  &\partial^2_{\tau}\mop{E}^{in}(\x+\v\tau,\tau)d\tau d\bi{x}d\bi{v}.
  \end{aligned}\label{eq:EscMT}
\end{equation}
For a typical ground moving target, the distance covered by the moving target while the
transmitted signal travels from the transmitter to the target and then from the target to the
receiver is much
smaller than the distance from the target to the receiver. Under this assumption,  we make the
following first-order approximation:
\begin{equation}
  \abs{\x+\v t-\y}\approx\abs{\x-\y}+\widehat{(\x-\y)}\cdot\v t\label{eq:taylor}
\end{equation}
where $\hat{\x}$ denotes the unit vector in the direction of $\x$.
Plugging~\eqref{eq:taylor} into the phase in~\eqref{eq:dyadicG} we have\footnote{In~\eqref{eq:Gvt}, we made the approximation $\abs{\x+\v t-\y}\approx\abs{\x-\y}$ for the denominator.}
\begin{equation}
  \begin{aligned}
  G(\x+\v t,\y,t) \approx \int &\left(I +
    \frac{\nabla\nabla}{k^2}\right)\frac{\mrm{e}^{\mrm{i}\omega\abs{\mathbf{y}-\mathbf{x}}/c_0}}{4\pi\abs{\mathbf{y}-\mathbf{x}}}\\
  &\mrm{e}^{\mrm{i}\omega (1+ \widehat{(\x-\y)}\cdot\v/c_0)t}d\omega.
  \end{aligned}\label{eq:Gvt}
\end{equation}
Thus, plugging~\eqref{eq:Gvt} into the inverse Fourier transform of~\eqref{eq:Ei}, we have
\begin{equation}
  \mop{E}^{in}(\x+\v\tau,\tau)\approx
  \mop{E}^{in}(\x,(1+\widehat{(\x-\y)}\cdot\v/c_0)\tau).\label{eq:Einappx}
\end{equation}
Now, using~\eqref{eq:Einappx} and~\eqref{eq:Gvt} in~\eqref{eq:EscMT},
the scattered field at the $k$-th receiver becomes
\begin{equation}
  \begin{aligned}
  \mop{E}^{sc}(\a_k^r,t) =
  \int &G(\a_k^r,\x,t-\alpha_k^r(\bi x,\bi
  v)\tau)\tilde{\mop{Q}}(\bi x,\bi v)\\
  &\partial_{\tau}^2\mop{E}^{in}(\x,\alpha^t(\bi x,\bi v)\tau)d\tau
  d\bi x d\bi v
  \end{aligned}\label{eq:EscMT2}
\end{equation}
where
\begin{align}
  \alpha_k^r(\bi x,\bi v) &= 1 -
  \frac{\widehat{(\x-\a_k^r)}\cdot\v}{c_0}\label{eq:alphar}\\
  \alpha^t(\bi x,\bi v) &= 1+ \frac{\widehat{(\x-\a^t)}\cdot\v}{c_0}.\label{eq:alphat}
\end{align}
Making the change of variables, $\tau' =
\alpha_k^r(\bi x,\bi v)\tau$ we have
\begin{equation}
  \begin{aligned}
  \mop{E}^{sc}(\a_k^r,t) =
  \int &G(\a_k^r,\x,t-\tau')\tilde{\mop{Q}}(\bi x,\bi v)\alpha_k^r(\bi
  x,\bi v)\\
  &\partial_{\tau'}^2\mop{E}^{in}(\x,\alpha_k(\bi x,\bi v)\tau')d\tau'
  d\bi x d\bi v
  \end{aligned}\label{eq:EscMT3}
\end{equation}
where
\begin{equation}
  \alpha_k(\bi x,\bi v) = \frac{\alpha^t(\bi x,\bi
    v)}{\alpha_k^r(\bi x,\bi v)}.
\end{equation}
In temporal Fourier domain,~\eqref{eq:EscMT3} becomes
\begin{equation}
  \begin{aligned}
  \mop{E}^{sc}(\a_k^r,\omega) = \int
  &G(\a_k^r,\x,\omega)\frac{\omega^2}{\alpha_k(\bi x,\bi
    v)^3}\tilde{\mop{Q}}(\bi x,\bi v)\\
  &\mop{E}^{in}\left(\x,\frac{\omega}{\alpha_k(\bi x,\bi v)}\right)
  d\bi x d\bi v.
  \end{aligned}\label{eq:EscMT4}
\end{equation}
For any realistic ground moving target, we can safely assume that $\abs{\v}\ll c_0$.  Thus, we make
the approximation that $\alpha^3_k\approx 1$ so that~\eqref{eq:EscMT4} becomes\footnote{By Cauchy-Schwarz we have that
  $\hat{\x}\cdot\v\leq \abs{\v}$.  Thus, from~\eqref{eq:alphar} and~\eqref{eq:alphat}, we
  have that $(\alpha_k^r)^3\leq 1+O(\frac{\abs{\v}}{c_0})$ and similar for $\alpha^t$.  Along
  with the assumption that $\abs{\v}\ll c_0$ we have $\alpha_k^3\approx 1$.}
\begin{equation}
  \begin{aligned}
  \mop{E}^{sc}(\a_k^r,\omega) = \omega^2\int
  &G(\a_k^r,\x,\omega)\tilde{\mop{Q}}(\bi x,\bi v)\\
  &\mop{E}^{in}\left(\x,\frac{\omega}{\alpha_k(\bi x,\bi v)}\right)
  d\bi x d\bi v.
  \end{aligned}\label{eq:EscMT5}
\end{equation}

\section{Dipole Target Model}\label{sec:dipoletarg}
By reciprocity and the fact that $\epot$ is assumed to be real-valued,
$\epot(\bm{x})$ admits an
eigendecomposition.  Let $\rho_1$, $\rho_2$, $\rho_3$ be the eigenvalues of $\epot$ and
$\e_1$, $\e_2$, $\e_3$ be the corresponding eigenvectors.  Then, the $\tilde{\mop{Q}}\mop{E}^{in}$ term
in~\eqref{eq:EscMT5} can be written as
\begin{equation}
  \tilde{\mop{Q}}\mop{E}^{in} = \sum_{i=1}^3 \varphi\mu_0\rho_i\left(\e_i^T\mop{E}^{in}\right)\e_i\label{eq:eigdecomp}.
\end{equation}
Note that we have suppressed $\bm{x}$ and $\bm{v}$ dependencies for $\rho_i$, $\e_i$, and $\varphi$  for notational simplicity.
The decomposition in~\eqref{eq:eigdecomp} can be interpreted as a source with $3$ colocated dipole antennas,
each with dipole moment $\e_i$~\cite{gustafsson04}.

\begin{assumption}
We approximate~\eqref{eq:eigdecomp} with the largest eigenvalue and the corresponding eigenvector.
\end{assumption}
Let $\rho$ be the largest eigenvalue and $\e_{sc}$ be the corresponding
eigenvector.  Then,
\begin{equation}
  \epot(\bm{x}) \approx \rho(\bm{x})\e_{sc}(\bm{x})\e_{sc}^T(\bm{x}).
\end{equation}
This approximation is equivalent to modeling the scatterer at
$\x\in\mathbb{R}^3$
by a short dipole antenna.  Under this model, an extended target becomes a
collection of dipole antennas.  Such a dipole model for scatterers was first studied by
Gustafsson in~\cite{gustafsson04} and more recently by Voccola
et. al. in~\cite{voccola13,voccola11}.
 Fig.~\ref{fig:dipole}
illustrates the idea of the dipole target model.
\begin{figure}[!ht]
  \centering
  \includegraphics[width=2.5in]{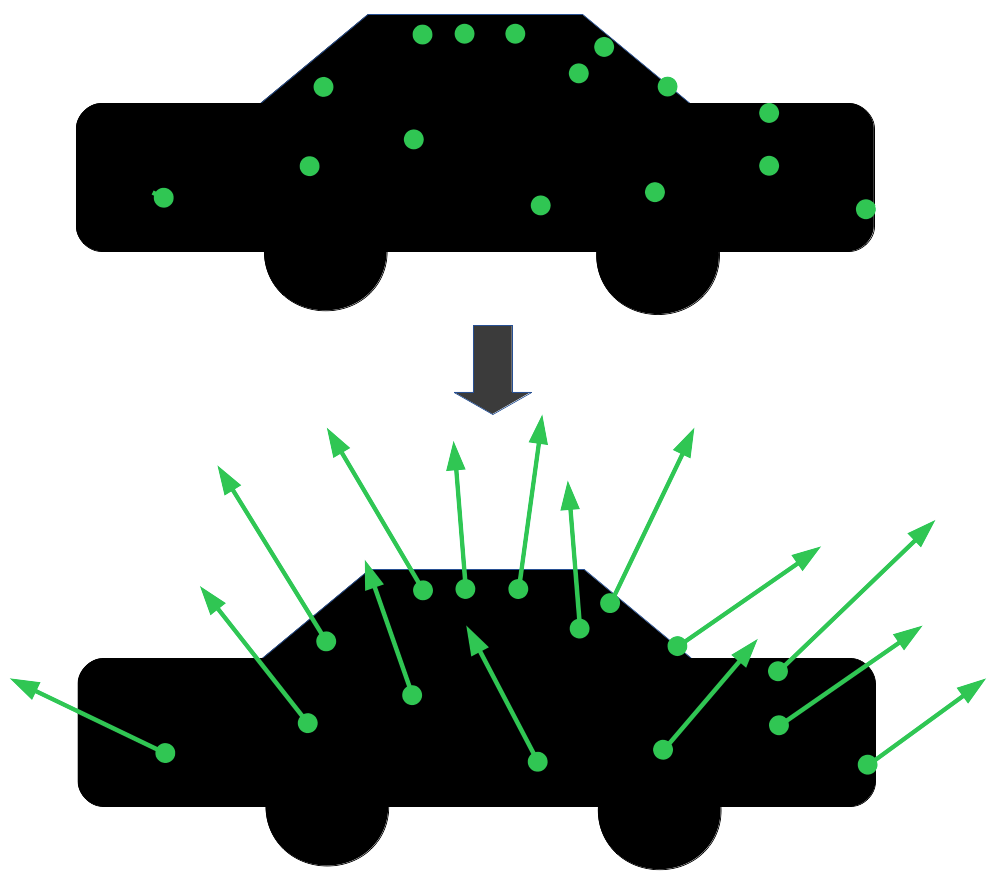}
  \caption{Dipole model of target.  Target is modeled by a collection of dipole
    antennas in lieu of traditional model as collection of points scatterers.}\label{fig:dipole}
\end{figure}
\section{Polarimetric Data Model}\label{sec:data_model}
In the previous sections, we derived the models for the incident and the scattered fields from a
moving target and the dipole target model approximation.  In this section, we utilize these models
to derive a model for
the signal received by polarimetrically diverse dipole receivers due to a signal originating from a
dipole transmitter and scattered by a moving target.
\subsection{Model for the Incident Field from a Short Dipole Antenna}
Let $\e^t$ denote the dipole moment of the transmit antenna.  Under the assumption that the length of the antenna is small compare to the
distance from the antenna to the scene, we can approximate~\eqref{eq:Ei} as
(see~\cite{voccola11,voccola13,son15})
\begin{equation}
  \mop{E}^{in}(\x,\omega) = \mrm{e}^{\mrm{i}\omega\abs{\x-\a^t}/c_0}A_t(\bi x,\omega)\hat{p}(\omega)\r^{\perp}_t(\bi x)\label{eq:Ein}
\end{equation}
where $\hat{p}(\omega)$ is the Fourier transform of the transmitted waveform $p(t)$, $A_t$ is
transmitter related terms including the
antenna beam patterns and the geometric spreading factors and
\begin{equation}
  \r^{\perp}_t(\bi x) = -\widehat{(\x - \a^t)}\times \widehat{(\x - \a^t)}\times \e^t.\label{eq:rperpt}
\end{equation}
Again, $\hat{\x}$ denotes the unit vector in the direction of $\x$ and $\r^{\perp}_t$ is a triple
vector cross product term that expresses the vector $\e^t$ projected onto the plane
with normal $\widehat{(\x-\a^t)}$. 


\subsection{The Model for the Ideal Target-Path Signal}
Let $\e_{k,H}^r$ and  $\e_{k,V}^r$ denote the dipole moments of the antennas at the $k$-th
receiver where $H$ and $V$ denote the horizontal and vertical dipole moments, respectively. Under
the assumption that the lengths of the dipole antennas are small as compared to the distance from
the
target to the receiver, the received signal at the $k$-th receiver can be modeled as~\cite{voccola11,voccola13,son15}
\begin{equation}
  \begin{aligned}
  \bi d_k^{0,TP}(\omega) = \omega^2\int &\mrm{e}^{\mrm{i}\omega
    \phi_{k}(\bi x,\bi v)/c_0}A_{t}\left(\bi
    x,\frac{\omega}{\alpha_k(\bi x,\bi
      v)}\right)\\
  &\hat{p}\left(\frac{\omega}{\alpha_k(\bi x,\bi v)}\right)
  \A_k^r(\bi x,\omega)\tilde{\bm{q}_{\e}}(\bi x,\bi v) d\bi x d\bi v,
  \end{aligned}\label{eq:dk}
\end{equation}
where
\begin{align}
  \tilde{\bm{q}_{\e}}(\bi x,\bi v) &=
  \rho(\bi x)\varphi(\bi v-\bi v_{\x})\e_{sc}(\bi x)\ip{\r^{\perp}_t(\bi x)}{\e_{sc}(\bi x)}\\
  \phi_{k}(\bi x,\bi v) &= \abs{\x-\a_k^r} +
  \frac{\abs{\x-\a^t}}{\alpha_k(\bi x,\bi v)}\label{eq:phi}\\
  \A_k^r(\bi x,\omega) &= [A_{r,k,H}(\bi
  x,\omega)\r_{r,k,H}^{\perp}(\bi x),A_{r,k,V}(\bi
  x,\omega)\r_{r,k,V}^{\perp}(\bi x)]^T\label{eq:Akr}\\
  \r_{r,k,p}^{\perp}(\bi x) &=
  -\widehat{(\x-\a_k^r)}\times
  \widehat{(\x-\a_k^r)}\times\e_{k,p}^r\quad p\in\{H,V\}.
\end{align}
Fig.~\ref{fig:receiveddata} illustrates the process of receiving the scattered signal from a dipole
transmitter, dipole target, and dipole receiver.  Note that the received signal depends on
$\r^{\perp}_{t}$, $\r^{\perp}_{r,k,p}$, and $\e_{sc}$.  More specifically, the strengths of the
signals are directly proportional to the scalar product between $\r^{\perp}_t$ and $\e_{sc}$ and
$\r^{\perp}_{r,k,p}$ and $\e_{sc}$.  This implies that the dipole-target model captures the
anisotropic scattering, because $\r^{\perp}_{r,k,p}$ depends on the look direction
$\widehat{(\x-\a_k^r)}$.  If $\e_{sc}$ is roughly parallel to $\widehat{(\x-\a_k^r)}$,
the received signal becomes weak.  Conversely, if $\e_{sc}$ is roughly orthogonal to
$\widehat{(\x-\a_k^r)}$, the received signal becomes strong.
\begin{figure}
  \centering
  \includegraphics[width=4in]{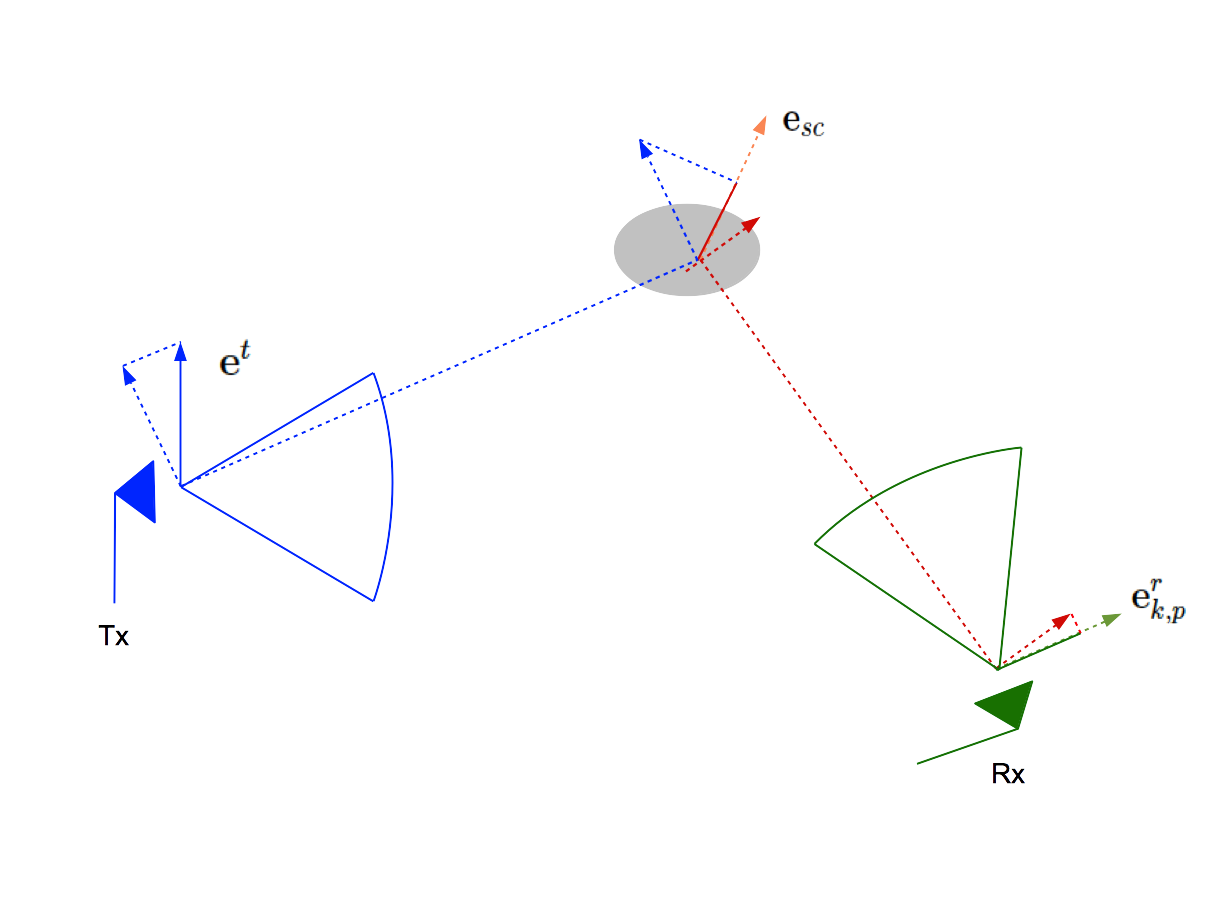}
  \caption{An illustration of the transmission of the electric field from a dipole transmitter,
    scattered by a
    dipole target, and then received by a dipole receiver.  Note that the strength of the scattered
    signal not only depends on the orientation of the dipole moments of the receiver, transmitter,
    and scatterer, but also the look directions from the transmitter to the scatterer and from
    the scatterer to
    the receiver.  This implies that the the strength of the scattered signal depends on the
    direction of the underlying field.  Hence, the dipole model captures anisotropic
    scattering.}\label{fig:receiveddata}
\end{figure}

The $A_{r,k,p}$ are terms related to the receive antennas including the
antenna beam patterns and the geometric spreading factors.
We
assume $A_{t}$ and $A_{r,k,p}$ are slowly varying functions of frequency, $\omega$.
Furthermore, we assume a narrowband (w.r.t. the center frequency, $\omega_0$) transmitted waveform
having frequency support
$\omega\in\Omega_0:=[\omega_0-\frac{B}{2},\omega_0+\frac{B}{2}]$.  Let
$p(t)=\tilde{p}(t)e^{i\omega_0 t}$ where $\tilde{p}(t)$ is a slowly varying function of $t$ such that
the approximation $\tilde{p}(\alpha_k t)\approx \tilde{p}(t)$ holds.  Thus,
$\hat{\tilde{p}}\left(\frac{\omega}{\alpha_k}\right) \approx \alpha_k\hat{\tilde{p}}(\omega)$.
Since $\hat{p}(\omega) = \hat{\tilde{p}}(\omega-\omega_0)$, we have
$\hat{p}\left(\frac{\omega}{\alpha_k}\right) \approx
\alpha_k\hat{\tilde{p}}(\omega-\alpha_k\omega_0)$ with bandlimitted support $\omega\in\Omega_k:=
[\alpha_k\omega_0-\frac{B}{2}, \alpha_k\omega_0+\frac{B}{2}]$, $k=1,\dots,M$.

\begin{assumption}
  We assume that $A_t$ and $A_{r,k,p}$ do not vary much within $\Omega_k$, and
  approximate $A_t$ and $A_{r,k,p}$ as constants w.r.t. $\omega$.\footnote{We justify this assumption
    by noting that most sources of opportunities are narrowband and antenna beampatterns are close
    to
    uniform w.r.t. $\omega$ within the band.}  
\end{assumption}

With this assumption,~\eqref{eq:dk} becomes
\begin{equation}
  \begin{aligned}
  \bi d_k^{0,TP}(\omega)\approx \omega_0^2\int &\mrm{e}^{\mrm{i}(\omega + \alpha_k\omega_0)
    \phi_{k}(\bi x,\bi v)/c_0}A_{t}(\bi
  x)\hat{\tilde{p}}(\omega)\\
  &\A_k^r(\bi x)\tilde{\bm{q}}_{\e}(\bi x,\bi v) d\bi x d\bi v,
  \end{aligned}\label{eq:dk2}
\end{equation}
where we make the substitution $\omega = \omega'-\alpha_k\omega_0$, $\omega'\in\Omega_k$ so that
$\omega\in\Omega_B:=[-\frac{B}{2},\frac{B}{2}]$.  Furthermore, we make the approximation for the scalar multiple
$(\omega+\alpha_k\omega_0)^2\approx \omega_0^2$ under the assumption that $\abs{\v}\ll c_0$ and
$B\ll c_0$.



\subsection{Model for the Ideal Direct-Path Signal}
The direct-path signal received at the $k$-th
receiver is proportional to the scalar product of the incident field given in~\eqref{eq:Ein} with the
dipole
moment of the receive
antenna, and does not include the target related terms present in the target-path signal defined in~\eqref{eq:dk2}.
The direct-path signal is also scaled by terms related to the receiver denoted by
$A^{DP}_{r,k,H}$ and $A^{DP}_{r,k,V}$.  Note that these terms may differ from
$A_{r,k,H}$ and $A_{r,k,V}$ defined in~\eqref{eq:dk} (for instance, $A_{r,k,H}$ and $A_{r,k,V}$
include a geometric spreading factor from target to the receiver that
$A_{r,k,H}^{DP}$ and $A_{r,k,V}^{DP}$ do not).  At the $k$-th receiver, the transmitter related term, $A_t(\x)$
in~\eqref{eq:Ein} is evaluated at the receiver location, $\x = \a_k^r$.  Let $A^{dp}_{t,k} =
A_t(\a_k^r)$, $k=1,\dots, M$.  Then, the ideal, noise-free direct-path signal at the $k$-th receiver becomes
\begin{equation}
  \bi d_{k}^{0,DP}(\omega) = \mrm{e}^{\mrm{i}(\omega+\omega_0)\abs{\a_k^r -
      \a^t}/c_0}A_{t,k}^{DP}\hat{\tilde{p}}(\omega)\A_{k,DP}^r\e^t,\quad \omega\in\Omega_B\label{eq:dp}
\end{equation}
where
\begin{equation}
  \A_{k,DP}^r = [A_{r,k,H}^{DP}\r_{r,k,H}^{\perp}(\a^t),A_{r,k,V}^{DP}\r_{r,k,V}^{\perp}(\a^t)]^T.
\end{equation}
Note that in~\eqref{eq:dp} we use the vector identity $\a\cdot(\b\times\c) = \c\cdot(\a\times\b) =
\b\cdot(\c\times\a)$ to make the substitution $\r_t^{\perp}(\a_k^r)\cdot \e_{k,p}^r =
\r_{r,k,p}^{\perp}(\a^t)\cdot \e^t$ where $p\in\{H,V\}$.

\subsection{Model for the Noisy Received Signals}\label{sec:noise_model}
\eqref{eq:dk2} and~\eqref{eq:dp} represent ideal, noise-free received signal models.  In practice,
however, all received signals are corrupted by noise. Thus, we model target-path and direct-path
received signals as
\begin{align}
  {\bm d}_k^{TP}(\omega) &= {\bm d}_k^{0,TP}(\omega) + {\bm n}_k^{TP}(\omega)\\
  {\bm d}_k^{DP}(\omega) &= {\bm d}_k^{0,DP} + {\bm n}_k^{DP}(\omega).
\end{align}
where
\begin{equation}
  {\bm n}_k^{TP}(\omega) = \left[ n_{k,H}^{TP}(\omega), n_{k,V}^{TP}(\omega)\right]^T
\end{equation}
and
\begin{equation}
  {\bm n}_k^{DP}(\omega) = \left[ n_{k,H}^{DP}(\omega), n_{k,V}^{DP}(\omega)\right]^T
\end{equation}
are the random noise processes present in the target-path and direct-path signals, respectively, at
the $k$-th receive antennas.

\begin{assumption}\label{as:noise}
  We assume that for all $\omega, \omega'$, ${\bm
    n}_k^{TP}(\omega)$ is independent of ${\bm n}_l^{TP}(\omega')$ for $k\neq l$ and that $n_{k,H}^{TP}(\omega)$ is
  independent of $n_{k,V}^{TP}(\omega')$.  Similarly, for all $\omega, \omega'$, ${\bm
    n}_k^{DP}(\omega)$ is independent of ${\bm n}_l^{DP}(\omega')$ for $k\neq l$ and  $n_{k,H}^{DP}(\omega)$ is
  independent of $n_{k,V}^{DP}(\omega')$.  Furthermore, ${\bm n}_k^{TP}(\omega)$ is
  independent of ${\bm n}_l^{DP}(\omega')$ for all $k,l,\omega,\omega'$.
  Finally, we assume that all noise processes are zero-mean Gaussian white noise processes,
  i.e., the auto-covariances functions $E[n_{k,p}^{TP}(\omega)(n^{TP}_{k,p}(\omega'))^*] =
  (\sigma_{k,p}^{TP})^2\delta(\omega-\omega')$
  and $E[n_{k,p}^{DP}(\omega)(n_{k,p}^{DP}(\omega'))^*] =
  (\sigma_{k,p}^{DP})^2\delta(\omega-\omega')$ for $k=1,...,M$ and $p\in\{H,V\}$.
\end{assumption}


We now consider the second-order statistics of all measurements.
These statistics can be succinctly expressed as matrix-valued covariance functions by first
considering the following random vectors:
\begin{equation}
\n^{TP}(\omega) = [n_{1,H}^{TP}(\omega),n_{1,V}^{TP}(\omega),\dots,n_{M,H}^{TP}(\omega),n_{M,V}^{TP}(\omega)]^T
\end{equation}
and
\begin{equation}
  \n^{DP}(\omega) = [n^{DP}_{1,H}(\omega),n^{DP}_{1,V}(\omega),\dots,n^{DP}_{M,H}(\omega),
  n^{DP}_{M,V}(\omega)]^T.
\end{equation}
Next, we define the following matrix valued covariance functions
\begin{equation}
  \vect{K}^{TP}(\omega,\omega') = E[\n^{TP}(\omega) (\n^{TP}(\omega'))^H]
\end{equation}
and
\begin{equation}
  \vect{K}^{DP}(\omega,\omega') = E[\n^{DP}(\omega) (\n^{DP}(\omega'))^H].
\end{equation}
By Assumption~\ref{as:noise} on the noise processes, $\vect{K}^{TP}(\omega,\omega')$ and
$\vect{K}^{DP}(\omega,\omega')$ can be expressed as
\begin{equation}
  \vect{K}^{TP}(\omega,\omega') = \bm{\Sigma}^{TP}\delta(\omega-\omega')
\end{equation}
and
\begin{equation}
  \vect{K}^{DP}(\omega,\omega') = \bm{\Sigma}^{DP}\delta(\omega-\omega').
\end{equation}
where $\bm{\Sigma}^{TP}$ and $\bm{\Sigma}^{DP}$ are diagonal matrices that do not depend on
$\omega$.

Finally, we assume that $\bm{\Sigma}^{TP}$ and $\bm{\Sigma}^{DP}$
are non-singular. Since $\bm{\Sigma}^{TP}$ and $\bm{\Sigma}^{DP}$ are both
diagonal matrices, this assumption is satisfied so long as no $n_{k,p}^{TP}(\omega)$ or $n_{k,p}^{DP}(\omega)$
is identically zero.



\section{Target Detection and Target Dipole Estimation}\label{sec:glrt}
We first set up a binary hypothesis test and next address it via generalized likelihood
ratio test (GLRT).  The solution to the GLRT leads to both the test-statistic for the detection task
as well as a method for
estimating the dipole moment of the target. 

\subsection{Target Detection}
We consider two detection problems: one where the direct-path signal is available and another
problem where
it is not.  We present a unified framework to address both problems.

For the rest of this section we assume that we have a
single point target.  For multiple targets, if we assume that scattered signal from each target can
be separated in time, i.e., the targets are sparsely distributed in space so that their respective
ranges to each receiver are far apart, then we can apply time-domain windows to the data and treat each
return separately as single point targets.

For
a point target located at
$\bm{x}^*$ and moving with velocity $\bm{v}^*$, $\tilde{\bm{q}}_{\e}(\bi x, \bi v) =
\rho^*\e_{sc}^*\delta(\bi x - \bi x^*)\delta(\bi v - \bi v^*)$ where
$\rho^*$ and $\e_{sc}^*$ are constants.
Thus,~\eqref{eq:dk2} becomes
\begin{equation}
  \bi d_k^{0,TP}(\omega) = \omega_0^2\mrm{e}^{\mrm{i}(\omega+\alpha_k^*(\bi x^*,\bi v^*)\omega_0)
    \phi_{k}(\bi x^*,\bi v^*)/c_0}A_{t}\hat{\tilde{p}}(\omega)
  \A_k^r\tilde{\bm{q}_{\e}}^*\label{eq:dkpt}
\end{equation}
where $\tilde{\bm{q}_{\e}}^* = \beta^*\rho^*\e_{sc}^*$ and $\beta^* =
\varphi(0)\ip{\r^{\perp}_t(\bi x^*)}{\e_{sc}^*}$.




We hypothesize that the point target is located at $\bi x_h$ moving with velocity $\bi v_h$ and let
$\overline{{\bm x}}_h = (\bi x_h, \bi v_h)$.
We consider the vectorized measurements
\begin{equation}
  \d^{TP} = [\bi d_1^{TP}, \bi d_2^{TP},\dots, \bi d_{M}^{TP}]^T
\end{equation}
and
\begin{equation}
  \d^{DP} = [\bi d^{DP}_1, \bi d^{DP}_2,\dots,\bi d^{DP}_{M}]^T.
\end{equation}
Now, let
\begin{equation}
  \D^{TP} = \opn{diag}[\mrm{e}^{\mrm{i}(\omega+\alpha_k(\overline{\bi x}_h))\phi_k(\overline{\bi x}_h)} I_2]\qquad k=1,\dots,M\label{eq:D_}
\end{equation}
and
\begin{equation}
  \D^{DP} = \opn{diag}[\mrm{e}^{\mrm{i}(\omega+\omega_0)\abs{\a_k^r-\a^t}/c_0} I_2],\qquad
  k=1,\dots,M\label{eq:Ddp_}
\end{equation}
where $\opn{diag}$ denotes a block diagonal matrix with $I_2$ being a
$2\times 2$ identity matrix.

Let
\begin{equation}
  \begin{aligned}
    \s^{DP} = [&A_{1,t}^{DP}(\A_{1,DP}^r\e^t)^T,A_{2,t}^{DP}(\A_{2,DP}^r\e^t)^T,\\
    &\dots,
  A_{M,t}^{DP}(\A_{M,DP}^r\e^t)^T]^T,
  \end{aligned}\label{eq:sDP_}
\end{equation}
\begin{equation}
  \bm{r}_{\e} = \omega_0^2A_t
                [(\A_1^r\tilde{\bm{q}}_{\e}^*)^T,(\A_2^r\tilde{\bm{q}}_{\e}^*)^T,\dots,(\A_{M}\tilde{\bm{q}}_{\e}^*)^T]^T.\label{eq:re_}
\end{equation}
In addition to target detection, we are interested in estimating $\bm{r}_{\e}$ up to a scaling
factor
to eventually estimate the target dipole, $\e_{sc}$.  Note that the transmitter term $A_t$ is included in
$\bm{r}_{\e}$.  However, since $A_t$ is merely a scaling term, it neither affects the solution to
the
target detection problem nor estimation of the target dipole.

Given~\eqref{eq:dp},~\eqref{eq:dkpt},\eqref{eq:D_},~\eqref{eq:Ddp_},~\eqref{eq:sDP_},
and~\eqref{eq:re_}, for each
hypothesized location and velocity pair, $\overline{\bi x}_h = ({\bi x}_h,{\bi v}_h)$, we address
the detection problem as
the following test of hypothesis:
\begin{equation}
  \begin{aligned}
    \mop{H}_1 &: \d^{TP}(\omega|\overline{\bi x}_h) = \hat{\tilde{p}}\D^{TP}\bm{r}_{\e} +
    \n^{TP}(\omega)\\
    &  \quad\d^{DP}(\omega) = \hat{\tilde{p}}\D^{DP}\s^{DP} + \n^{DP}(\omega)\\
    \mop{H}_0 &: \d(\omega|\overline{\bi x}_h) = \n^{TP}(\omega)\\
    &  \quad \d^{DP}(\omega) = \hat{\tilde{p}}\D^{DP}\s^{DP} + \n^{DP}(\omega).
  \end{aligned}\label{eq:hypothesis2}
\end{equation}
In~\eqref{eq:hypothesis2}, we assume that direct-path received signals are available.  If they are
not, we can omit $\d^{DP}(\omega)$ in the formulation or, equivalently, set them equal to zero.

Since both $\hat{\tilde{p}}$ and $\bm{r}_{\e}$ are unknown, we address the hypothesis
test in~\eqref{eq:hypothesis2} within the GLRT framework and write the test-statistic as follows:
\begin{equation}
 \lambda(\overline{\bi x}) =
  \frac{\underset{B_t,\bm{r}_{\e}}{\max}\ dP_{\mop{H}_1|\overline{\bi
        x}}}{\underset{B_t,\bm{r}_{\e}}{\max}\ dP_{\mop{H}_0|\overline{\bi x}}}\label{eq:glrt}
\end{equation}
where $P_{\mop{H}_1|\overline{\bi x}}$ and $P_{\mop{H}_0|\overline{\bi x}}$ are probability
measures for the stochastic processes under $\mop{H}_1$ and $\mop{H}_0$, respectively, given
$\overline{\bi x}$.  Thus, the GLRT is the ratio of the maximum likelihood estimates (MLEs) of $\tilde{\hat{p}}$
and $\bm{r}_{\e}$ using the
Radon-Nicodym derivatives of
$P_{\mop{H}_1|\overline{\bi x}}$ and $P_{\mop{H}_0|\overline{\bi x}}$.  Under the noise assumption
in Assumption~\ref{as:noise}, we have Gaussian pdfs for each $\omega$ under $\mop{H}_1$ and
$\mop{H}_0$.

The solution to~\eqref{eq:glrt} is given in Theorem~\ref{thm:binhyp}.
\begin{theorem}[Direct-path signal available.]\label{thm:binhyp}
  Suppose $\hat{\tilde{p}}$ is continuous in $\Omega_B$ and $\abs{\rho^*}>0$.
  Let
  \begin{equation}
    \tilde{\d} = [(\d^{TP})^T, (\d^{DP})^T]^T,
  \end{equation}
  \begin{equation}
    \tilde{\vect{D}} = \begin{bmatrix}
      \vect{D}^{TP} & \bm{0}\\
      \bm{0} & \vect{D}^{DP}
    \end{bmatrix},
  \end{equation}
  and
  \begin{equation}
    \tilde{\bm{\Sigma}} = \begin{bmatrix}
      \bm{\Sigma}^{TP} & \bm{0}\\
      \bm{0} & \bm{\Sigma}^{DP}
    \end{bmatrix}.
  \end{equation}
  Furthermore, let
  \begin{equation}
    \vect{Q}_1 = \int
      \tilde{\vect{D}}^H\tilde{\d}\tilde{\d}^H\tilde{\vect{D}}
      d\omega,\label{eq:Q1}
  \end{equation}
  and
  \begin{equation}
    \vect{Q}_0 = \int (\vect{D}^{DP})^H\d^{DP}(\d^{DP})^H\vect{D}^{DP}
      d\omega.\label{eq:Q0}
  \end{equation}
  Then, under Assumption~\ref{as:noise}, the optimal test statistics to the GLRT
  problem~\eqref{eq:glrt} is
  \begin{equation}
    \begin{aligned}
    \lambda(\overline{\bi x}) =
    \lambda_{\max}(&\tilde{\bm{\Sigma}}^{-1/2}\vect{Q}_1\tilde{\bm{\Sigma}}^{-1/2})\\
    &- \lambda_{\max}((\bm{\Sigma}^{DP})^{-1/2}\vect{Q}_0(\bm{\Sigma}^{DP})^{-1/2}).
    \end{aligned}\label{eq:lmda1}
  \end{equation}
\end{theorem}
\begin{proof}
  See Appendix~\ref{apdx:binhyp}.
\end{proof}

If the direct-path signal is unavailable, e.g. because it is not collected,
we set $\d^{DP} \equiv \bm{0}$ in the binary hypothesis
problem~\eqref{eq:hypothesis2}.  The solution to this modified problem is given by the
  following corollary to Theorem~\ref{thm:binhyp}.
\begin{corollary}[Direct-path signal is not available]\label{cor:binhyp}
  By letting $\d^{DP}\equiv\bm{0}$, the test statistic~\eqref{eq:lmda1} becomes
  \begin{equation}
    \lambda(\overline{\bi x}) =
    \lambda_{\max}\left((\bm{\Sigma}^{TP})^{-1/2}\R(\bm{\Sigma}^{TP})^{-1/2}\right).\label{eq:lmda2}
  \end{equation}
  where
  \begin{equation}
    \R = \int (\D^{TP})^H\d^{TP}(\d^{TP})^H\D^{TP} d\omega.
  \end{equation}
\end{corollary}
\begin{proof}
  See Appendix~\ref{apdx:cor}.
\end{proof}

\subsection{Estimation of Dipole Direction}
Let $\w$ be the eigenvector of $\tilde{\bm{\Sigma}}^{-1/2}\vect{Q}_1\tilde{\bm{\Sigma}}^{-1/2}$ corresponding to
its largest eigenvalue.   By~\eqref{eq:reEig}, we see that $\w$ is the estimate of
$\tilde{\bm{\Sigma}}^{-1/2}\tilde{\s}$ up to a scalar.  Now, since
$\tilde{\bm{\Sigma}}$ is block diagonal,
\begin{equation}
  \tilde{\bm{\Sigma}}^{-1/2} = \begin{bmatrix}
    \bm{\Sigma}^{-1/2} & \bm{0}\\
    \bm{0} & (\bm{\Sigma}^{DP})^{-1/2}
  \end{bmatrix}.\label{eq:Sig-12}
\end{equation}
Then, by~\eqref{eq:Sig-12} and~\eqref{eq:stilde}, $\w$ can be written as
\begin{equation}
  \w = [\w_1^T, \w_2^T]^T.
\end{equation}
where $\w_1$ is the estimate of $\bm{\Sigma}^{-1/2}\bm{r}_{\e}$ up to a scalar.
Hence, by the definition of
$\bm{r}_{\e}$ in~\eqref{eq:re_}, it follows that $\e_{sc}$ can be approximated by unit vector in
the direction of
$\A_r^{+}\bm{\Sigma}^{1/2}\w_1$ where $\A_r^{+}$ is the Moore-Penrose pseudoinverse
of
\begin{equation}
  \A_r = [(\A_1^r)^T,(\A_2^r)^T,\dots,(\A_{M}^r)^T]^T.
\end{equation}

When no direct-path signal is available, by Corollary~\ref{cor:binhyp}, we can restrict our
attention to eigenvector corresponding to the maximum eigenvalue of
$(\bm{\Sigma}^{TP})^{-1/2}\R(\bm{\Sigma}^{TP})^{-1/2}$.  Let $\u$ be that eigenvector.  Then, by similar
reasoning as above, $\u$ is the estimate of $\bm{\Sigma}^{-1/2}\bm{r}_{\e}$.  Thus, we can again
approximate $\e_{sc}$ by the unit vector in the direction of $\A_r^{+}(\bm{\Sigma}^{TP})^{1/2}\u$.

\section{Numerical Simulations}\label{sec:simulations}
We evaluate the performance of the GLRT-based spatially resolved detection and
dipole
moment estimation using simulated data.
Specifically, we compare the performance between the set up with and without polarimetric
diversity.  The effect of polarimetric diversity depends on the spatial location of the
antennas.  We further evaluate the effect of spatial diversity on the detection performance with
and without polarimetric diversity. We also compare the performance due to polarimetric
diversity with and without direct-path signal.

\subsection{The Configuration for Simulations}
The configuration we use in all simulation is as follows:  A
single
transmitter is located at $(15,15,5)$ km from the scene center transmitting an $8$ MHz bandwidth
signal.
The scene is a $400\times 400$ m flat ground topography.
All receivers lie on a circle of radius $10$ km from the scene's center on the ground plane, all at $5$ km above
the scene.

The transmitted waveforms are linearly polarized in a single direction with center frequency of $2$
GHz.  The dipole moment of the
transmit antenna is parallel to the ground plane in all cases.  Two linearly polarized dipole
antennas, labeled $H$ and $V$ are
placed at each receiver.
The dipole moment of each receiver's $H$-polarized antenna is parallel to the ground and points in
the direction $(\sin \theta_n, -\cos\theta_n,0)$, where $\theta_n$ is the azimuth angle of the
$n$-th receiver location relative to the $x$-axis.  For the $V$-polarized receive antennas, the
dipole moment directions all point upwards in the $z$-axis direction, i.e. $(0,0,1)$.
Only the $H$-polarized receive antennas are used for the case in which there is no polarization
diversity present in the receivers.
Fig.~\ref{fig:configsim} illustrates the simulation setup.
\begin{figure}
  \centering
  \includegraphics[width=0.45\textwidth]{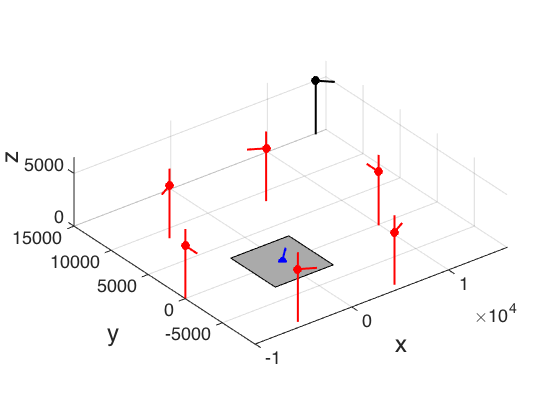}
  \caption{An illustration of the configuration of the antennas and the scene.  The transmitter is
    located at position $(15,15,5)$ km.  The receivers are positioned around a circle of radius
    $10$ km.}\label{fig:configsim}
\end{figure}

Fig.~\ref{fig:3targ_img} shows an example test-statistic image for $3$ point targets.  The plot
also shows the dipole moment directions of each target (solid blue line at each target) and its
estimated dipole moment (yellow dashed line at each target).  Average target-path signal SNR and
average direct-path signal SNR were both kept at $0$ dB, and $6$ receivers were used with $256$
samples.
\begin{figure}[!ht]
  \centering
  \includegraphics[width=0.475\textwidth]{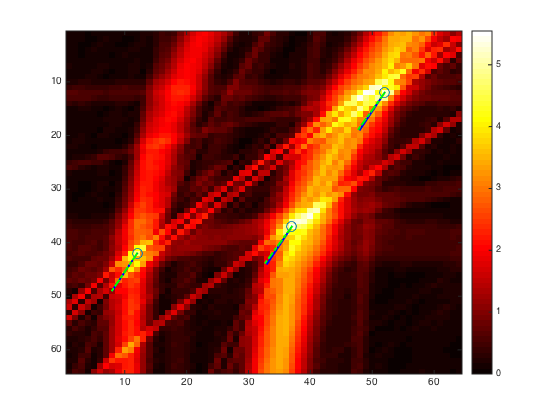}
  \caption{Example test-statistic image for $3$ point targets.  The target locations are indicated
    by the circle.  The true dipole moment direction
    is shown in blue solid line and estimated dipole moment direction is shown in yellow dashed
    line.}\label{fig:3targ_img}
\end{figure}
\subsection{Description of the Experiements}
The numerical simulations presented in this section is categorized into two cases:
\begin{enumerate}
\item The effect of polarization diversity on target detection task.
\item The effect of polarization diversity on target dipole moment estimation.
\end{enumerate}
We study the detection task via probability of detection as our metric.  We present three subsets of
simulated experiments for the detection task.
\begin{enumerate}
\item  No spatial diversity: Single receiver case with direct-path signals.
\item  Limited spatial diversity: Two receiver case with and without direct-path signals.
\item  Full spatial diversity: Six receivers surrounding the scene.
\end{enumerate}
In first two scenarios, we limit the spatial diversity to focus on the effect of polarization
diversity.  The first scenario presents a case with single receiver in which we compare case with
polarization diversity and without where direct-path signal is available in both cases.  In order
to compare with the case without direct-path signal, in the second scenario we place two receivers
about the scene.  Lastly, we examine the case with full spatial diversity and evaluate
the detection performance therein. In the first two scenarios, objective is to show polarization
diversity can improve detection performance when spatial diversity is limited.  The third scenario,
 we aim to show that
polarization diversity can help to fully exploit spatial diversity.

\subsection{Target Detection Performance}
To evaluate the target detection performance using the GLRT-based approach stated in
Theorem~\ref{thm:binhyp}, we numerically estimate the probability of detection $P_d$ under constant
false alarm rate (CFAR) of $0.001$.  The threshold values that achieve the CFAR are determined
numerically using $10,000$ realizations of noise under the null hypothesis, $\mathcal{H}_0$.  The
$\overline{\bm x}$ is held constant for each experiment.  The noise processes added to each
received signal are all assumed to have common variances of $(\sigma^{TP})^2$ for target-path signals and
$(\sigma^{DP})^2$ for direct-path signals.  In other words, we set $\bm{\Sigma}^{TP} = (\sigma^{TP})^2I$ and
$\bm{\Sigma}^{DP} = (\sigma^{DP})^2 I$ where $I$ is the identity matrix.  The variances are
determined by the average signal SNRs.

\subsubsection{Single receiver with direct-path signal}
We begin our investigation with the simplest scenario of a single receiver with direct-path signal
available.  In this scenario, we only compare cases with direct-path signal between polarization
diversity and no polarization diversity.  We look at two cases, with high direct-path signal SNR
(10 dB)
and with low direct-path signal SNR (-30 dB). The receiver was placed so that $H$-polarized and
$V$-polarized target-path signals had roughly equivalent SNRs.  The simulation results are shown in
Fig.~\ref{fig:pds_v_SNR_1rec}.
\begin{figure}[!ht]
  \centering
  \begin{subfigure}[b]{0.45\textwidth}
    \includegraphics[width=\textwidth]{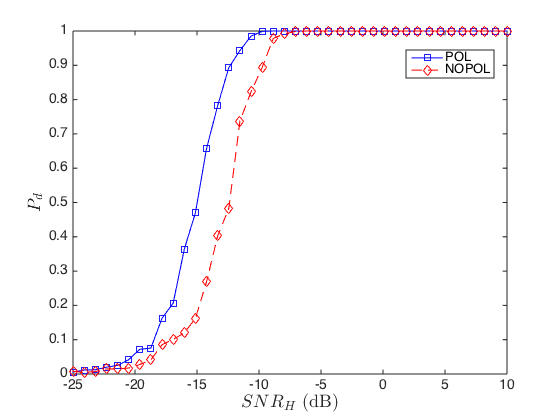}
    \caption{Direct-path SNR : 10 dB}\label{fig:pds_v_SNR_1rec_a}
  \end{subfigure}
  \begin{subfigure}[b]{0.45\textwidth}
    \includegraphics[width=\textwidth]{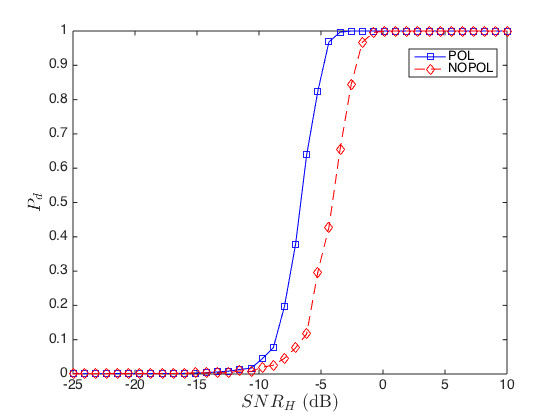}
    \caption{Direct-path SNR : -30 dB}\label{fig:pds_v_SNR_1rec_b}
  \end{subfigure}
  \caption{The probability of detection ($P_d$) vs. target-path signal SNR ($SNR_{H}$) for a single
    receiver scenario.  Blue line labeled POL is with polarization diversity.  The red dashed line
    labeled NOPOL is without polarization diversity (with only $H$-polarized signal being used for
    detection).
    $P_d$ was estimated using $10,000$ realizations of noise under
    CFAR of $0.001$.  In both cases, at fixed $P_d=0.9$, there is an improvement of $2.63$ dB in
    signal SNR.}\label{fig:pds_v_SNR_1rec}
\end{figure}

In Fig.~\ref{fig:pds_v_SNR_1rec}, we see improvement in detection performance for polarimetrically
diverse case over no polarimetric diversity, with both high direct-path SNR and low direct-path
SNR.  The improvement is similar in both cases.  The high direct-path SNR primarily functions to
shift the probability of detection graph.  It also changes the slope of the probability of
detection graph so that the degradation in probability of detection as function of target-path
signal SNR is more gradual.  Using fixed $P_d=0.9$ as a figure of merit, we see
a $2.63$ dB improvement in received signal SNR.  If we fix $P_d=0.5$ for case without polarimetric
diversity, we see that there is improvement of $0.4$ in probability of detection with polarimetric
diversity for high direct-path SNR case and improvement of $0.5$ in probability of detection for
low direct-path SNR case.

\subsubsection{Two receiver with and without direct-path signal}
In this subsection, we present simulations with limited spatial diversity, with two receivers.  By
having two receivers, we can in edition compare scenarios between case with direct-path signals
available and case without.  We again look at case with high direct-path SNR ($10$ dB) and low
direct-path SNR ($-30$ dB).  Fig.~\ref{fig:pds_v_SNR_2rec} shows the result.
\begin{figure}[!ht]
  \centering
  \begin{subfigure}[b]{0.45\textwidth}
    \includegraphics[width=\textwidth]{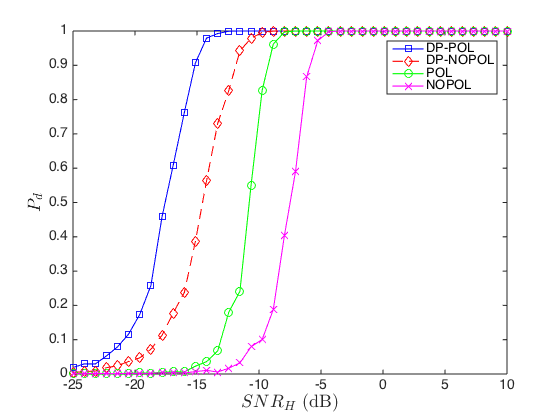}
    \caption{Average direct-path signal SNR : 10 dB}
  \end{subfigure}
  \begin{subfigure}[b]{0.45\textwidth}
    \includegraphics[width=\textwidth]{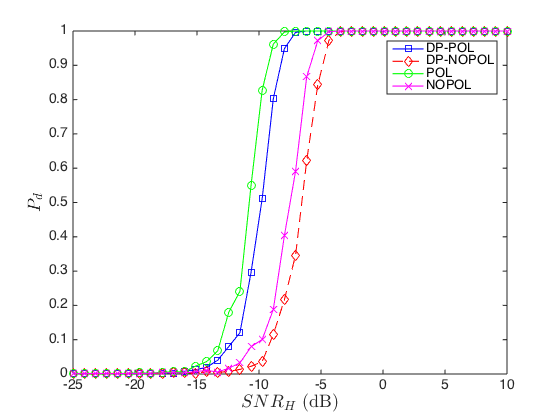}
    \caption{Average direct-path signal SNR : -30 dB}
  \end{subfigure}
  \caption{The probability of detection ($P_d$) vs. average target-path signal SNR ($SNR_{AVG}$)
    for a two receiver scenario.  Blue line labeled DP-POL is with polarization diversity and
    direct-path signals available.  The red dashed line
    labeled DP-NOPOL is without polarization diversity (with only $H$-polarized signal being used for
    detection) and direct-path signals available.  Green line labeled POL is with polarization
    diversity and without direct-path signals available.  Magenta line labeled NOPOL is without
    polarization diversity and without direct-path signals available.
    $P_d$ was estimated using $10,000$ realizations of noise under
    CFAR of $0.001$. At fixed $P_d=0.9$, we see $3.31$ dB improvement for high average direct-path
    SNR in
    average target-path signal
    SNRs with polarization diversity over without.  With low average direct-path SNR, the
    improvement is $3.38$ dB.  Without direct-path, the improvement is $3.39$ dB.}\label{fig:pds_v_SNR_2rec}
\end{figure}

We see improvement in detection performance for polarimetrically
diverse case over no polarimetric diversity, with both high direct-path SNR and low direct-path
SNR.  The improvement is again similar in both cases.  For fixed $P_d=0.9$ as a figure of merit, we see
a $3.31$ and $3.38$ dB improvement in average received signal SNR for case with high average direct-path SNR and
low average direct-path SNR, respectively.  Without direct-path, the improvement is $3.39$ dB.  If we fix $P_d=0.5$ for case without polarimetric
diversity, we see that there is improvement of about $0.5$ in probability of detection with polarimetric
diversity for both high direct-path SNR, low direct-path SNR, and without direct-path signal cases.
Furthermore, we see that with low direct-path SNR, the detection performance degrades to be
comparable to, in fact, a bit worse than that of case without direct-path signals available.

\subsubsection{Full spatial diversity: $6$ receivers}
The final scenario we examined was with full spatial diversity using $6$ receivers equally spaced
on the circle around the scene.  For this experiment, we show results for $3$ different average direct-path
signal SNRs.  Fig.~\ref{fig:pds_v_SNR_6rec} shows the results of the probability
of detection the $3$ different average direct-path SNR cases.  Comparing
Fig.~\ref{fig:pds_v_SNR_6rec_a} and Fig.\ref{fig:pds_v_SNR_6rec}, we see that the improvement in the
probability of detection graph is minimal with growing high direct-path SNRs.
\begin{figure*}
  \centering
    \begin{subfigure}[b]{0.45\textwidth}
    \includegraphics[width=\textwidth]{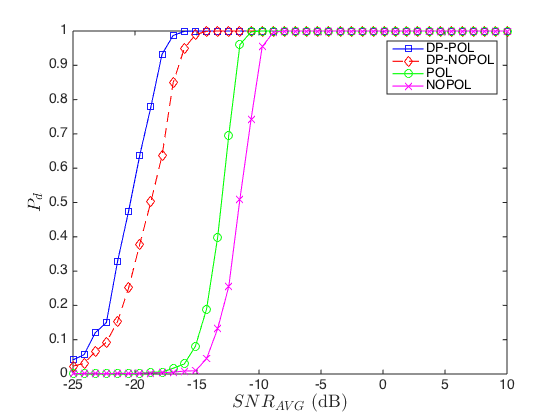}
    \caption{Average direct-path SNR : $10$ dB}\label{fig:pds_v_SNR_6rec_a}
  \end{subfigure}
  \begin{subfigure}[b]{0.45\textwidth}
    \includegraphics[width=\textwidth]{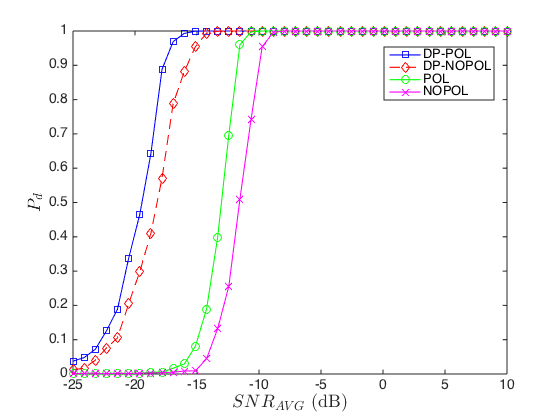}
    \caption{Average direct-path SNR : $0$ dB}\label{fig:pds_v_SNR_6rec_b}
  \end{subfigure}
  \begin{subfigure}[b]{0.45\textwidth}
    \includegraphics[width=\textwidth]{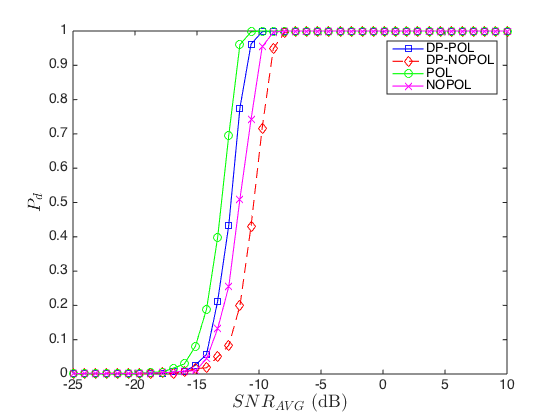}
    \caption{Average direct-path SNR : $-30$ dB}\label{fig:pds_v_SNR_6rec_c}
  \end{subfigure}
  \caption{The probability of detection ($P_d$) vs. average target-path signal SNR ($SNR_{AVG}$)
    for a two receiver scenario.  Blue line labeled DP-POL is with polarization diversity and
    direct-path signals available.  The red dashed line
    labeled DP-NOPOL is without polarization diversity (with only $H$-polarized signal being used for
    detection) and direct-path signals available.  Green line labeled POL is with polarization
    diversity and without direct-path signals available.  Magenta line labeled NOPOL is without
    polarization diversity and without direct-path signals available.
    $P_d$ was estimated using $10,000$ realizations of noise under
    CFAR of $0.001$.  At fixed $P_d=0.9$, we see $1.54$ dB improvement for high average direct-path
    SNR in
    average target-path signal
    SNRs with polarization diversity over without.  With low average direct-path SNR, the
    improvement is $1.90$ dB.  For the middle case, the improvement is $1.89$ dB.  Without direct-path, the improvement is $1.77$ dB.}\label{fig:pds_v_SNR_6rec}
\end{figure*}

For fixed $P_d=0.9$ as a figure of merit, we see
a $1.59$ dB improvement in average received signal SNR is observed for high average direct-path SNR ($10$
dB), $1.89$ dB improvement with average direct-path SNR at $0$ dB and $1.90$ dB for low average
direct-path SNR ($-30$ dB).  Without direct-path, the improvement is $1.77$ dB.  Fixing $P_d=0.5$ for case without polarimetric
diversity, we see that there is improvement of about $0.27$ in probability of detection with polarimetric
diversity for $10$ dB average direct-path SNR, $0.31$ for $0$ dB average direct-path SNR, and $0.5$
for $-30$ dB average direct-path SNR.  For case without direct-path the improvement in $P_d$ is
$0.45$.

\subsection{Numerical Simulation of Target Dipole Moment Estimation}
We also examine the performance of dipole moment estimation. We
again use
$6$ receivers equally spaced on a circle about the center of the scene.  The performance criteria
used is the the angle between the true dipole moment and the estimated
one.  Specifically, we use $1,000$ realizations of noise and use the sample average of the
estimated angles,
labeled $\Delta\phi$, in radians, between the
true dipole moment and the estimated one as the final performance
criteria. We repeated this process for various target-path noise levels.  Fig.~\ref{fig:dipoleest} shows
the resulting graph of $\Delta\phi$ vs. average target-path signal SNRs ($SNR_{AVG}$).  Note that non-diverse cases completely fails
to accurately estimate the dipole moment.  Namely, $\Delta\phi$ is close to $\pi/2$ for all
$SNR_{AVG}$
values.  This is due to the fact that for the non-diverse case, estimated dipole moment all have a
dominant component along the $z$-axis.  The $z$-axis component of the target dipole moment cannot be
estimated accurately since it lies in the orthogonal complement of the subspace spanned by the dipole moments
of the $H$-polarized receive antennas as they all lie in a plane parallel to
the ground plane.  Furthermore, in the case of dipole moment estimation, high level of direct-path
noise degrades the estimation performance to the extent that having only the target-path signal results in
better estimates. This is due to the fact that, although the direct-path signals do not contain
information about the dipole moment of the target directly, the estimation of
$\tilde{\s}$, i.e. the
eigenvector, does depend on the direct-path signals.
\begin{figure}[!ht]
  \centering
  \begin{subfigure}[b]{0.45\textwidth}
    \includegraphics[width=\textwidth]{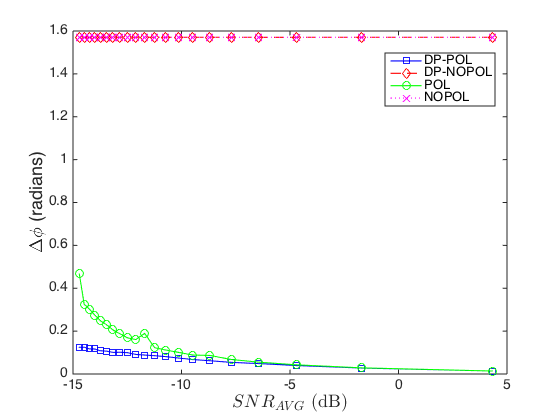}
    \caption{Average direct-path SNR : $4$ dB}
  \end{subfigure}
  \begin{subfigure}[b]{0.45\textwidth}
    \includegraphics[width=\textwidth]{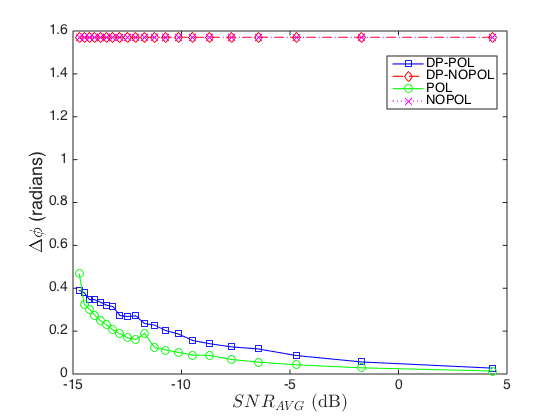}
    \caption{Average direct-path SNR : $-13$ dB}
  \end{subfigure}
  \caption{Average absolute angle between estimated dipole moment direction of the target and the
    true dipole moment ($\Delta\phi$) vs. average target-path SNR ($SNR_{AVG}$) at two different
    average direct-path SNRs.  Blue solid line and red dashed line
    (DP-POL and DP-NOPOL) are cases with
    direct-path signals while green solid line and magenta dotted line (POL and NOPOL) are case
    without direct-path signals.  DP-POL and POL are cases with polarimetric diversity and DP-NOPOL
    and NOPOL are cases without polarimetric diversity.  One
    thousand realizations of
    noise was used to estimate $\Delta\phi$. $6$ equally space receivers on a circle of radius 10km
    from the center of the scene were used for processing.}\label{fig:dipoleest}
\end{figure}



\section{Conclusions and Future Works}\label{sec:conclusions}
% ...
In this paper, we derived from first principles, a novel polarimetric data model for moving target that
takes into account anisotropic scattering.  This model represents a target as a spatially distributed collection of dipole antennas.  We considered a bistatic
scenario with a single transmitter of opportunity equipped with a dipole antenna with unknown dipole
moment direction and multiple receivers equipped with a pairs
of orthogonally polarized dipole antennas.
We address the target detection with and without direct-path signal in a unified GLRT
framework and derived a method to estimate the dipole moment of the target.  The detection
test-statistic 
is given in terms of the maximum eigenvalues of whitened data
correlation matrices.  For the case with the direct-path signal, the test-statistic is the
difference between the
maximum eigenvalues of full data correlation matrix (with both target-path and direct-path received
signals) and the correlation matrix for the direct-path signals only.  When direct-path signals are not
available, the test-statistic reduces to the maximum eigenvalue of the target-path correlation
matrix.  The dipole moment estimation was derived from the eigenvector associated with the maximum
eigenvalue of the data correlation matrix.
This estimation method falls
out naturally from the solution to a generalized eigenvalue problem from the GLRT
framework.

In addition, through series of numerical simulations, we show that polarimetric diversity helps in
both the detection and dipole estimation tasks.  Specifically, numerical simulations show that polarimetric diversity improves the probability of detection under constant false alarm
rate.  The improvement was observed both with limited spatial diversity and full spatial diversity.
For the dipole
estimation, we further show that, without diversity, full $3$-dimensional estimation of the
dipole moment is untenable due to non-trivial null space. We further showed that availability of
direct-path signal  generally improves both the detection and the estimation tasks.

In the companion paper~\cite{son17analysis}, we provide a mathematical analysis of the detection
performance. 
For further extension of this work, we may add structured noise such as clutter and examine the detection
test-statistics under such scenario.  An alternative to GLRT could also be employed for the
detection task, such as a Bayesian framework.



\bibliographystyle{IEEEtran}
\bibliography{IEEEabrv,citations}

\begin{thebibliography}{10}
\providecommand{\url}[1]{#1}
\csname url@samestyle\endcsname
\providecommand{\newblock}{\relax}
\providecommand{\bibinfo}[2]{#2}
\providecommand{\BIBentrySTDinterwordspacing}{\spaceskip=0pt\relax}
\providecommand{\BIBentryALTinterwordstretchfactor}{4}
\providecommand{\BIBentryALTinterwordspacing}{\spaceskip=\fontdimen2\font plus
\BIBentryALTinterwordstretchfactor\fontdimen3\font minus
  \fontdimen4\font\relax}
\providecommand{\BIBforeignlanguage}[2]{{%
\expandafter\ifx\csname l@#1\endcsname\relax
\typeout{** WARNING: IEEEtran.bst: No hyphenation pattern has been}%
\typeout{** loaded for the language `#1'. Using the pattern for}%
\typeout{** the default language instead.}%
\else
\language=\csname l@#1\endcsname
\fi
#2}}
\providecommand{\BIBdecl}{\relax}
\BIBdecl

\bibitem{wang09}
Q.~Wang, Y.~Lu, and C.~Hou, ``An experimental wimax based passive radar
  study,'' in \emph{Microwave Conference, 2009. APMC 2009. Asia Pacific}, Dec
  2009, pp. 1204--1207.

\bibitem{wang10}
L.~Wang, I.-Y. Son, and B.~Yazici, ``Passive imaging using distributed
  apertures in multiple-scattering environments,'' \emph{Inverse Problems},
  vol.~26, 2010.

\bibitem{wang12}
\BIBentryALTinterwordspacing
L.~Wang and B.~Yazici, ``Passive imaging of moving targets exploiting multiple
  scattering using sparse distributed apertures,'' \emph{Inverse Problems},
  vol.~28, no.~12, p. 125009, 2012. [Online]. Available:
  \url{http://stacks.iop.org/0266-5611/28/i=12/a=125009}
\BIBentrySTDinterwordspacing

\bibitem{LWang12}
------, ``Passive imaging of moving targets using sparse distributed
  apertures,'' \emph{SIAM Journal on Imaging Sciences}, vol.~5, no.~3, pp.
  769--808, 2012.

\bibitem{Wang11}
L.~Wang, C.~Yarman, and B.~Yazici, ``Doppler-hitchhiker: A novel passive
  synthetic aperture radar using ultranarrowband sources of opportunity,''
  \emph{Geoscience and Remote Sensing, IEEE Transactions on}, vol.~49, no.~10,
  pp. 3521--3537, Oct 2011.

\bibitem{Baker05_1}
H.~Griffiths and C.~Baker, ``Passive coherent location radar systems. part 1:
  performance prediction,'' \emph{Radar, Sonar and Navigation, IEE Proceedings
  -}, vol. 152, no.~3, pp. 153--159, June 2005.

\bibitem{baker05_2}
C.~Baker, H.~Griffiths, and I.~Papoutsis, ``Passive coherent location radar
  systems. part 2: waveform properties,'' \emph{Radar, Sonar and Navigation,
  IEE Proceedings -}, vol. 152, no.~3, pp. 160--168, June 2005.

\bibitem{Dawidowicz12}
B.~Dawidowicz, P.~Samczynski, M.~Malanowski, J.~Misiurewicz, and K.~Kulpa,
  ``Detection of moving targets with multichannel airborne passive radar,''
  \emph{Aerospace and Electronic Systems Magazine, IEEE}, vol.~27, no.~11, pp.
  42--49, November 2012.

\bibitem{Kulpa12}
P.~Krysik and K.~Kulpa, ``The use of a gsm-based passive radar for sea target
  detection,'' in \emph{Radar Conference (EuRAD), 2012 9th European}, Oct 2012,
  pp. 142--145.

\bibitem{Wacks14}
S.~Wacks and B.~Yazici, ``Passive synthetic aperture hitchhiker imaging of
  ground moving targets - part 1: Image formation and velocity estimation,''
  \emph{Image Processing, IEEE Transactions on}, vol.~23, no.~6, pp.
  2487--2500, June 2014.

\bibitem{Yarman10}
\BIBentryALTinterwordspacing
C.~E. Yarman, L.~Wang, and B.~Yaziciı, ``Doppler synthetic aperture hitchhiker
  imaging,'' \emph{Inverse Problems}, vol.~26, no.~6, p. 065006, 2010.
  [Online]. Available: \url{http://stacks.iop.org/0266-5611/26/i=6/a=065006}
\BIBentrySTDinterwordspacing

\bibitem{mason15}
E.~Mason, I.-Y. Son, and B.~Yazici, ``Passive synthetic aperture radar imaging
  using low-rank matrix recovery methods,'' \emph{IEEE Journal of Selected
  Topics in Signal Processing}, 2015, to appear.

\bibitem{son2007radar}
I.-Y. Son, T.~Varslot, C.~E. Yarman, A.~Pezeshki, B.~Yazici, and M.~Cheney,
  ``Radar detection using sparsely distributed apertures in urban
  environment,'' in \emph{Defense and Security Symposium}.\hskip 1em plus 0.5em
  minus 0.4em\relax International Society for Optics and Photonics, 2007, pp.
  65\,671Q--65\,671Q.

\bibitem{hack12}
D.~Hack, L.~Patton, A.~Kerrick, and M.~Saville, ``Direct cartesian detection,
  localization, and de-ghosting for passive multistatic radar,'' in
  \emph{Sensor Array and Multichannel Signal Processing Workshop (SAM), 2012
  IEEE 7th}, June 2012, pp. 45--48.

\bibitem{hack14}
D.~E. Hack, L.~K. Patton, B.~Himed, and M.~A. Saville, ``Detection in passive
  {MIMO} radar networks,'' \emph{IEEE Transactions on Signal Processing},
  vol.~62, no.~11, pp. 2999--3012, June 2014.

\bibitem{hack2014centralized}
D.~E. Hack, L.~K. Patton, B.~Himed, M.~Saville \emph{et~al.}, ``Centralized
  passive mimo radar detection without direct-path reference signals,''
  \emph{Signal Processing, IEEE Transactions on}, vol.~62, no.~11, pp.
  3013--3023, 2014.

\bibitem{bialkowski2011generalized}
K.~S. Bialkowski, S.~D. Howard \emph{et~al.}, ``Generalized canonical
  correlation for passive multistatic radar detection,'' in \emph{Statistical
  Signal Processing Workshop (SSP), 2011 IEEE}.\hskip 1em plus 0.5em minus
  0.4em\relax IEEE, 2011, pp. 417--420.

\bibitem{palmer2013dvb}
J.~E. Palmer, H.~A. Harms, S.~J. Searle, and L.~M. Davis, ``Dvb-t passive radar
  signal processing,'' \emph{Signal Processing, IEEE Transactions on}, vol.~61,
  no.~8, pp. 2116--2126, 2013.

\bibitem{colone2011direction}
F.~Colone, G.~De~Leo, P.~Paglione, C.~Bongioanni, and P.~Lombardo, ``Direction
  of arrival estimation for multi-frequency fm-based passive bistatic radar,''
  in \emph{Radar Conference (RADAR), 2011 IEEE}.\hskip 1em plus 0.5em minus
  0.4em\relax IEEE, 2011, pp. 441--446.

\bibitem{cui14}
G.~Cui, J.~Liu, H.~Li, and B.~Himed, ``Target detection for passive radar with
  noisy reference channel,'' in \emph{Radar Conference, 2014 IEEE}, 2014, pp.
  0144--0148.

\bibitem{webster12}
T.~Webster, ``Scalar and vector multistatic radar data models,'' Ph.D.
  dissertation, Rensselaer Polytechnic Institute, Troy, NY, December 2012.

\bibitem{lee09}
J.-S. Lee and E.~Pottier, \emph{Polarimetric Radar Imaging: From Basics to
  Applications}, 1st~ed., ser. Optical Science and Engineering.\hskip 1em plus
  0.5em minus 0.4em\relax Boca Raton, FL: CRC Press, 2009, vol. 143.

\bibitem{cloude96}
S.~R. Cloude and E.~Pottier, ``A review of target decomposition theorems in
  radar polarimetry,'' \emph{IEEE Transactions on Geoscience and Remote
  Sensing}, vol.~34, no.~2, pp. 498--518, March 1996.

\bibitem{boerner07}
W.-M. Boerner, ``Basics of sar polarimetry i,'' DTIC Document, Tech. Rep.,
  2007.

\bibitem{moreira13}
A.~Moreira, P.~Prats-Iraola, M.~Younis, G.~Krieger, I.~Hajnsek, and
  K.~Papathanassiou, ``A tutorial on synthetic aperture radar,'' \emph{IEEE
  Geoscience and Remote Sensing}, vol.~1, no.~1, pp. 6--43, March 2013.

\bibitem{jackson09}
J.~A. Jackson, ``Three-dimensional feature models for synthetic aperture radar
  and experiments in feature extraction,'' Ph.D. dissertation, Ohio State
  University, 2009.

\bibitem{son17analysis}
I.-Y. Son and B.~Yazici, ``Performance analysis of passive polarimetric
  multistatic radar detection of moving targets,'' 2017, to be published.

\bibitem{sinclair48}
G.~Sinclair, ``Modification of the radar target equation for arbitrary targets
  and arbitrary polarization,'' Antenna Laboratory, The Ohio State Univsersity
  Research Foundation, Tech. Rep., 1948.

\bibitem{sinclair50}
------, ``The transmission and reception of elliptically polarized waves,''
  \emph{Proceedings of the IRE}, vol.~38, no.~2, pp. 148--151, 1950.

\bibitem{kennaugh52}
E.~M. Kennaugh, ``Polarization properties of radar reflections,'' Master's
  thesis, The Ohio State University, March 1952.

\bibitem{deschamps51}
G.~A. Deschamps, ``Geometrical representation of the polarization of a plane
  electromagnetic wave,'' \emph{Proceedings of the IRE}, vol.~39, no.~5, pp.
  540--544, 1951.

\bibitem{graves56}
C.~D. Graves, ``Radar polarization power scattering matrix,'' \emph{Proceedings
  of the IRE}, vol.~44, no.~2, pp. 248--252, 1956.

\bibitem{voccola13}
K.~Voccola, M.~Cheney, and B.~Yazici, ``Polarimetric synthetic-aperture
  inversion for extended targets in clutter,'' \emph{Inverse Problems},
  vol.~29, no.~5, April 2013.

\bibitem{voccola11}
K.~Voccola, ``Statistical and analytical techniques in synthetic aperture
  radar,'' Ph.D. dissertation, Rensselaer Polytechnic Institute, Troy, NY,
  August 2011.

\bibitem{webster2014b}
T.~Webster, M.~Cheney, and E.~L. Mokole, ``Multistatic polarimetric radar data
  modeling and imaging of moving targets,'' \emph{Inverse Problems}, vol.~30,
  no.~3, p. 035002, 2014.

\bibitem{gustafsson04}
M.~Gustafsson, ``Multi-static synthetic aperture radar and inverse
  scattering,'' \emph{Technical Report LUTEDX/(TEAT-7123)/1-28/(2003)}, 2004.

\bibitem{tyo06}
J.~S. Tyo, D.~L. Goldstein, D.~B. Chenault, and J.~A. Shaw, ``Review of passive
  imaging polarimetry for remote sensing applications,'' \emph{Applied Optics},
  vol.~45, no.~22, pp. 5453--5469, August 2006.

\bibitem{wapenaar13}
K.~Wapenaar and J.~Thorbecke, ``On the retrieval of the directional scattering
  matrix from directional noise,'' \emph{SIAM Journal of Imaging Science},
  vol.~6, pp. 322--340, 2013.

\bibitem{colone2014}
F.~Colone and P.~Lombardo, ``Exploiting polarimetric diversity in fm-based
  pcl,'' in \emph{Radar Conference (Radar), 2014 International}.\hskip 1em plus
  0.5em minus 0.4em\relax IEEE, 2014, pp. 1--6.

\bibitem{colone2015}
------, ``Polarimetric passive coherent location,'' \emph{Aerospace and
  Electronic Systems, IEEE Transactions on}, vol.~51, no.~2, pp. 1079--1097,
  2015.

\bibitem{son15}
I.-Y. Son and B.~Yazici, ``Passive imaging with multistatic polarimetric
  radar,'' in \emph{2015 IEEE International Radar Conference}, 2015.

\bibitem{chew99}
W.~C. Chew, \emph{Waves and Fields in Inhomogenous Media}.\hskip 1em plus 0.5em
  minus 0.4em\relax Wiley-IEEE Press, 1999.

\bibitem{colton98}
D.~Colton and R.~Kress, \emph{Inverse Acoustic and Electromagnetic Scattering
  Theory}, 2nd~ed.\hskip 1em plus 0.5em minus 0.4em\relax Springer, 1998.

\end{thebibliography}
\appendix
\subsection{Proof of Theorem~\ref{thm:binhyp}}\label{apdx:binhyp}
Define the inner product
\begin{equation}
  \ip{\mathbf{f}}{\mathbf{g}}_{\vect{K}} = \int \mathbf{g}^H\K^{-1}\mathbf{f} d\omega d\omega' = \int
  \vect{g}^H\bm{\Sigma}^{-1}\vect{f} d\omega
\end{equation}
and the associated norm as
\begin{equation}
  \norm{\mathbf{f}}_{\vect{K}} = \sqrt{\ip{\mathbf{f}}{\mathbf{f}}_{\vect{K}}} = \sqrt{\int\vect{f}^H\bm{\Sigma}^{-1}\vect{f} d\omega}.
\end{equation}

Under the Gaussian noise assumptions in Section~\ref{sec:noise_model},  we have the following
log-likelihood function under $\mop{H}_1$
\begin{equation}
  \begin{aligned}
    \ell_{\mop{H}_1}(\hat{\tilde{p}},\s^{DP},\bm{r}_{\e}) =  -\frac{1}{2} &\left(\norm{\vect{d} -
        \hat{\tilde{p}}\vect{D}\tilde{\bm{r}}_{\e}}_{\vect{K}^{TP}}^2 \right. \\
      & \left. + \norm{\d^{DP} -
        \hat{\tilde{p}}\D^{DP}\s^{DP}}_{\vect{K}^{DP}}^2\right),
    \end{aligned}\label{eq:glrtdpH1}
\end{equation}
and under $\mop{H}_0$
\begin{equation}
    \ell_{\mop{H}_0}(\hat{\tilde{p}},\s^{DP}) =  - \frac{1}{2}\left(\norm{\d}_{\vect{K}^{TP}}^2 + \norm{\d^{DP} -
        \hat{\tilde{p}}\D^{DP}\s^{DP}}_{\vect{K}^{DP}}^2\right).\label{eq:glrtdpH0}
\end{equation}

The~\eqref{eq:glrt} implies that we can then solve the MLE problem by maximizing~\eqref{eq:glrtdpH1} w.r.t. $(\hat{\tilde{p}},\s^{DP},\bm{r}_{\e})$
and maximizing~\eqref{eq:glrtdpH0} w.r.t. $(\hat{\tilde{p}},\s^{DP})$ and taking their difference.

Let
\begin{equation}
  \tilde{\s} = [\bm{r}_{\e}^T,(\s^{DP})^T]^T\label{eq:stilde}
\end{equation}
and
\begin{equation}
  \tilde{\vect{K}}(\omega,\omega') = \tilde{\bm{\Sigma}}\delta(\omega-\omega').
\end{equation}
Then~\eqref{eq:glrtdpH1} becomes
\begin{equation}
    \ell_{\mop{H}_1}(\hat{\tilde{p}},\tilde{\s}) =  -\frac{1}{2} \left(\norm{\tilde{\vect{d}} -
        \hat{\tilde{p}}\tilde{\vect{D}}\tilde{\s}}_{\tilde{\vect{K}}}^2\right).\label{eq:glrtdpH12}
\end{equation}

We maximize~\eqref{eq:glrtdpH12} first over $\hat{\tilde{p}}$ and find
$\hat{\tilde{p}}$ in terms
of $\tilde{\s}$, then
maximize it over $\tilde{\s}$.

By maximizing over $\hat{\tilde{p}}$, we arrive at a solution given by the following Lemma.
\begin{lemma}\label{lemma:p}
  Given~\eqref{eq:glrtdpH12},
  \begin{equation}
    \hat{\tilde{p}}^* = \underset{{\hat{\tilde{p}}}}{\opn{argmax}}\ \ell_{\mop{H}_1}(\hat{\tilde{p}},\tilde{\s})
    = \frac{\tilde{\s}^H\tilde{\bm{\Sigma}}^{-1}\tilde{\vect{D}}^H\tilde{\d}}
  {\tilde{\s}^H\tilde{\bm{\Sigma}}^{-1}\tilde{\s}}.\label{eq:AtMLE}
  \end{equation}
\end{lemma}
\begin{proof}
  See  Appendix~\ref{apdx:A}.
\end{proof}

Using Lemma~\ref{lemma:p}, we plug~\eqref{eq:AtMLE} into~\eqref{eq:glrtdpH12} and we arrive at
\begin{equation}
  \underset{{\hat{\tilde{p}}}}{\opn{max}}\ \ell_{\mop{H}_1}(\hat{\tilde{p}},\tilde{\s}) =
   \frac{1}{2}\left(J(\tilde{\s})-\norm{\tilde{\vect{d}}}_{\tilde{\vect{K}}}^2\right)
\end{equation}
where
\begin{equation}
  J(\tilde{\s}) = \int
    \frac{\abs{\tilde{\d}^H\tilde{\vect{D}}\tilde{\bm{\Sigma}}^{-1}\tilde{\s}}^2}
    {\tilde{\s}^H\tilde{\bm{\Sigma}}^{-1}\tilde{\s}} d\omega.\label{eq:J}
\end{equation}
In Appendix~\ref{apdx:B}, it is shown that maximizing~\eqref{eq:J} w.r.t. $\tilde{\s}$ is
equivalent to solving the following
generalized eigenvalue problem:
\begin{equation}
  \tilde{\bm{\Sigma}}^{-1}\vect{Q}_1\tilde{\bm{\Sigma}}^{-1}\tilde{\s} = J(\tilde{\s})\tilde{\bm{\Sigma}}^{-1}\tilde{\s}\label{eq:reEig}.
\end{equation}
The solution to~\eqref{eq:reEig} is given by
$\lambda_{\max}(\tilde{\bm{\Sigma}}^{-1/2}\vect{Q}_1\tilde{\bm{\Sigma}}^{-1/2})$.  Therefore,
\begin{equation}
  \underset{\hat{\tilde{p}},\tilde{\s}}{\max}\ \ell_{\mop{H}_1}(\hat{\tilde{p}},\tilde{\s}) = \frac{1}{2}
  \left(\lambda_{\max}(\tilde{\bm{\Sigma}}^{-1/2}\vect{Q}_1\tilde{\bm{\Sigma}}^{-1/2})
  - \norm{\tilde{\d}}_{\tilde{\vect{K}}}^2\right).\label{eq:l1mle}
\end{equation}
We can follow similar procedure to solve maximization of~\eqref{eq:glrtdpH0} and arrive at
\begin{equation}
  \begin{aligned}
  \underset{\hat{\tilde{p}},\s^{DP}}{\max}\ \ell_{\mop{H}_0}(\hat{\tilde{p}},\s^{DP}) = \frac{1}{2}
  &\left(\lambda_{\max}((\bm{\Sigma}^{DP})^{-1/2}\vect{Q}_0(\bm{\Sigma}^{DP})^{-1/2})\right.\\
    &\left.- \norm{\tilde{\d}}_{\tilde{\vect{K}}}^2\right).
  \end{aligned}\label{eq:l2mle}
\end{equation}
Noting that common scalar term $1/2$ is irrelevant, we arrive at our conclusion that
\begin{equation}
  \begin{aligned}
    \lambda(\overline{\bi x}) =
    \lambda_{\max}(&\tilde{\bm{\Sigma}}^{-1/2}\vect{Q}_1\tilde{\bm{\Sigma}}^{-1/2})\\ -
    &\lambda_{\max}((\bm{\Sigma}^{DP})^{-1/2}\vect{Q}_0(\bm{\Sigma}^{DP})^{-1/2}).
  \end{aligned}
\end{equation}

\subsection{Proof of Lemma~\ref{lemma:p}}\label{apdx:A}
Since~\eqref{eq:glrtdpH12} is Fr\'{e}chet differentiable we use G\^{a}teaux derivative to
differentiate~\eqref{eq:glrtdpH12} w.r.t $\hat{\tilde{p}}$.  First note that
\begin{equation}
  \underset{{\hat{\tilde{p}}}}{\opn{argmax}}\ \ell_{\mop{H}_1}(\hat{\tilde{p}},\tilde{\s}) =
  \underset{{\hat{\tilde{p}}}}{\opn{argmax}}\
  2\opn{Re}\{\ip{\tilde{\d}}{\hat{\tilde{p}}\tilde{\vect{D}}\tilde{\s}}_{\tilde{\vect{K}}}\}
    - \norm{\hat{\tilde{p}}\tilde{\vect{D}}\tilde{\s}}_{\tilde{\vect{K}}}^2
\end{equation}
By definition of G\^{a}teaux derivative we have that
\begin{align}
  D\left(\norm{\hat{\tilde{p}}\tilde{\vect{D}}\tilde{\s}}_{\tilde{\vect{K}}}^2\right)(h) &= D(\ip{\tilde{\vect{D}}\tilde{\s}\hat{\tilde{p}}}{\tilde{\vect{D}}\tilde{\s}\hat{\tilde{p}}}_{\tilde{\vect{K}}})(h)\\
  &= 2\opn{Re}\{\ip{\tilde{\vect{D}}\tilde{\s}\hat{\tilde{p}}}{\tilde{\vect{D}}\tilde{\s}h}_{\tilde{\vect{K}}}\}\\
  &= 2\opn{Re}\left\{\int \tilde{\s}^H\tilde{\vect{D}}^H\tilde{\bm{\Sigma}}^{-1}\tilde{\vect{D}}\tilde{\s}\hat{\tilde{p}}h^*
    d\omega\right\}\\
  &= 2\opn{Re}\left\{\int \tilde{\s}^H\tilde{\bm{\Sigma}}^{-1}\tilde{\s}\hat{\tilde{p}}h^*
    d\omega\right\}\\
  &= 2\opn{Re}\{\ip{\tilde{\s}\hat{\tilde{p}}}{\tilde{\s}h}_{\tilde{\vect{K}}}\}
\end{align}
where, in the last two lines, we used the fact that diagonal matrices commute and that $\tilde{\vect{D}}$ is
unitary for every $\omega$.
By similar computation, we have that
\begin{equation}
  \begin{aligned}
  D(\ip{\d}{\hat{\tilde{p}}\tilde{\vect{D}}\tilde{\s}})(h) &= \ip{\d}{h\tilde{\vect{D}}\tilde{\s}}_{\tilde{\vect{K}}} =
  \ip{\tilde{\vect{D}}^H\d}{h\tilde{\s}}_{\tilde{\vect{K}}}\\
  &= \int
  \tilde{\s}^H\tilde{\bm{\Sigma}}^{-1}\tilde{\vect{D}}^H\d h^* d\omega.
  \end{aligned}
\end{equation}
Thus, we have
\begin{equation}
  D\ell(\hat{\tilde{p}},\tilde{\s}) = \opn{Re}\{\ip{\tilde{\vect{D}}^H\d}{\tilde{\s}\cdot}_{\tilde{\vect{K}}} -
  \ip{\tilde{\s}\hat{\tilde{p}}}{\tilde{\s}\cdot}_{\tilde{\vect{K}}}\}.\label{eq:Dell}
\end{equation}
Setting~\eqref{eq:Dell} equal to zero implies the desired result.

\subsection{Derivation of~\eqref{eq:reEig}}\label{apdx:B}
Note that $\tilde{\s}$ is not a function of $\omega$.  Under the assumption that $\abs{\rho}^*>0$, the integrand of
$J(\tilde{\s})$ is Lipschitz w.r.t. $\tilde{\s}$.  
Furthermore, $\tilde{\vect{D}}$ is continuous w.r.t. $\omega$ and so is $\tilde{\d}$ given that $\hat{\tilde{p}}$ is continuous w.r.t. $\omega$.  Since the domain of integration,
$\Omega_B$, is compact, and the integrand continuous w.r.t. $\omega$, the integrand is Lipschitz w.r.t. $\omega$.
Thus, by Lebesgue Dominated Convergence
theorem we can interchange the integral and the gradient operator arriving at
\begin{equation}
  \nabla_{\tilde{\s}}J(\tilde{\s}) = \int \nabla_{\tilde{\s}} \left(\frac{\abs{\tilde{\d}^H\tilde{\vect{D}}\tilde{\bm{\Sigma}}^{-1}\tilde{\s}}^2}
    {\tilde{\s}^H\tilde{\bm{\Sigma}}^{-1}\tilde{\s}}\right) d\omega.\label{eq:nabJ1}
\end{equation}
We now use the product rule to evaluate the gradient of the integrand.  First, we have that
\begin{equation}
  \nabla_{\tilde{\s}}\abs{\tilde{\d}^H\tilde{\vect{D}}\tilde{\bm{\Sigma}}^{-1}\tilde{\s}}^2 =
  2\tilde{\bm{\Sigma}}^{-1}\tilde{\vect{D}}^H\tilde{\d}\tilde{\d}^H\tilde{\vect{D}}\tilde{\bm{\Sigma}}^{-1}\tilde{\s}.\label{eq:num}
\end{equation}
Next, we have by the chain rule,
\begin{equation}
  \nabla_{\tilde{\s}}\frac{1}{\tilde{\s}^H\tilde{\bm{\Sigma}}^{-1}\tilde{\s}} =
  \frac{-2\tilde{\bm{\Sigma}}^{-1}\tilde{\s}}{(\tilde{\s}^H\tilde{\bm{\Sigma}}^{-1}\tilde{\s})^2}.\label{eq:denom}
\end{equation}
Using~\eqref{eq:num} and~\eqref{eq:denom}, we have
\begin{equation}
  \begin{aligned}
    &\int\nabla_{\tilde{\s}}\left(\frac{\abs{\tilde{\d}^H\tilde{\vect{D}}\tilde{\bm{\Sigma}}^{-1}\tilde{\s}}^2}
      {\tilde{\s}^H\tilde{\bm{\Sigma}}^{-1}\tilde{\s}}\right) d\omega\\
    &=2\int\left(\frac{\tilde{\bm{\Sigma}}^{-1}\tilde{\vect{D}}^H\tilde{\d}\tilde{\d}^H\tilde{\vect{D}}\tilde{\bm{\Sigma}}^{-1}\tilde{\s}}{\tilde{\s}^H\tilde{\bm{\Sigma}}^{-1}\tilde{\s}}-
      \frac{\abs{\tilde{\d}^H\tilde{\vect{D}}\tilde{\bm{\Sigma}}^{-1}\tilde{\s}}^2\tilde{\bm{\Sigma}}^{-1}\tilde{\s}}{(\tilde{\s}^H\tilde{\bm{\Sigma}}^{-1}\tilde{\s})^2}\right)
    d\omega \\
    &=
    \frac{2}{\tilde{\s}^H\tilde{\bm{\Sigma}}^{-1}\tilde{\s}}\left(\tilde{\bm{\Sigma}}^{-1}\vect{Q}_1\tilde{\bm{\Sigma}}^{-1}-
        J(\tilde{\s})\tilde{\bm{\Sigma}}^{-1}\right)\tilde{\s}.\label{eq:nabJ2}
  \end{aligned}
\end{equation}
Plugging in~\eqref{eq:nabJ2} into~\eqref{eq:nabJ1} and setting $\nabla_{\tilde{\s}}J = 0$ we
arrive at the generalized eigenvalue problem in~\eqref{eq:reEig}.

\subsection{Proof of Corollary~\ref{cor:binhyp}}\label{apdx:cor}
By setting $\d^{DP}\equiv\bm{0}$, we
have that $\bm{\Sigma}^{DP} = \bm{0}$.  This leads to
\begin{equation}
  \tilde{\bm{\Sigma}} = \begin{bmatrix}
    \bm{\Sigma}^{TP} & \bm{0}\\
    \bm{0} &  \bm{0}
  \end{bmatrix}.
\end{equation}
Note that $\tilde{\bm{\Sigma}}$ is no longer invertible.  We substitute
$\tilde{\bm{\Sigma}}^{-1}$ by its Moore-Penrose pseudoinverse,
\begin{equation}
  \tilde{\bm{\Sigma}}^{+} = \begin{bmatrix}
    (\bm{\Sigma}^{TP})^{-1} & \bm{0}\\
    \bm{0} & \bm{0}
  \end{bmatrix}.
\end{equation}
This leads to replacing $\tilde{\bm{\Sigma}}^{-1/2}\vect{Q}_1\tilde{\bm{\Sigma}}^{-1/2}$ by
\begin{equation}
  (\tilde{\bm{\Sigma}}^{+})^{1/2}\vect{Q}_1(\tilde{\bm{\Sigma}}^{+})^{1/2}
\end{equation}
Let
\begin{equation}
  \R = \int (\D^{TP})^H\d^{TP}(\d^{TP})^H\D^{TP} d\omega.
\end{equation}
Then,
\begin{equation}
  (\tilde{\bm{\Sigma}}^{+})^{1/2}\vect{Q}_1(\tilde{\bm{\Sigma}}^{+})^{1/2} = \begin{bmatrix}
    (\bm{\Sigma}^{TP})^{-1/2}\R(\bm{\Sigma}^{TP})^{-1/2} & \bm{0}\\
    \bm{0} & \bm{0}
  \end{bmatrix}
\end{equation}
We also have that $\vect{Q}_0 = \bm{0}$ and $(\bm{\Sigma}^{DP})^{+} = \bm{0}$.  This leads to the
test statistic~\eqref{eq:lmda1}
in Theorem~\ref{thm:binhyp} being equivalent to~\eqref{eq:lmda2}.
\end{document}